\documentclass[conference]{IEEEtran}

\usepackage[utf8]{inputenc}
\usepackage[T1]{fontenc}
\usepackage{lmodern}
\usepackage[english]{babel}
\usepackage[kerning=true]{microtype}

\usepackage{fullpage}

\usepackage{xspace}
\usepackage{url}
\usepackage{todonotes}

\usepackage{amssymb}
\usepackage{amsfonts}
\usepackage{amsmath}
\usepackage{amsthm}
\usepackage{mathtools}
\usepackage[shortcuts]{extdash}

\usepackage{stmaryrd}

\usepackage{graphicx}

\usepackage{hyperref}
\hypersetup{hidelinks}
\usepackage{color}

\theoremstyle{plain}
\newtheorem{theorem}{Theorem}[section]
\newtheorem{lemma}[theorem]{Lemma}

\newtheorem{proposition}[theorem]{Proposition}

\theoremstyle{definition}

\newtheorem{example}[theorem]{Example}
\newtheorem{definition}[theorem]{Definition}
\newtheorem{remark}[theorem]{Remark}
\newtheorem{claim}[theorem]{Claim}

\newcommand{\added}[1]{#1}

\renewcommand{\restriction}{\mathord{\upharpoonright}}
\usepackage{amsmath}

\DeclareMathOperator*{\argmin}{arg\,min}

\newcommand{\set}[1]{\{#1\}}
\newcommand{\pr}{\mathbb{P}}

\newcommand{\st}{\ :\ }
\newcommand{\src}{\sigma}
\newcommand{\tgt}{\tau}

\newcommand{\source}[1]{\mathrm{source}(#1)}
\newcommand{\target}[1]{\mathrm{target}(#1)}

\newcommand{\inter}{\mathrm{Int}}

\newcommand{\nat}{\mathbb{N}}
\newcommand{\two}{\mathbf{2}}

\let\oldfrac\frac
\renewcommand{\frac}[2]{%
  \mathchoice
    {\oldfrac{#1}{#2}}
    {#1/#2}
    {#1/#2}
    {#1/#2}
}

\newcommand{\ie}{{\em i.e.}\xspace}
\newcommand{\eg}{{\em e.g.}\xspace}

\newcommand{\ap}{{\em a~priori}\xspace}

\newcommand{\mso}{{\sc mso}\xspace}
\newcommand{\msonab}{{\sc mso}+$\nabla$\xspace}
\newcommand{\msou}{{\sc mso+u}\xspace}

\newcommand{\bld}[1]{\ifmmode \mathbf{#1} \else \textbf{#1}\fi}
\newcommand{\io}{\bld{io}\xspace}
\newcommand{\fo}{\bld{fo}\xspace}

\newcommand{\itv}[1]{\mathcal{#1}}
\newcommand{\ifi}{\itv{I}}
\newcommand{\ifj}{\itv{K}}
\newcommand{\ifk}{\itv{J}}

\newcommand{\vs}[1]{\mathbf{#1}\xspace}
\newcommand{\vsf}{\vs{f}}
\newcommand{\vsg}{\vs{g}}
\newcommand{\vsh}{\vs{h}}

\newcommand{\suc}{\mathrm{Suc}}
\newcommand{\pre}{\mathrm{Pre}}

\newcommand{\length}{\mathrm{Len}}

\newcommand{\eqdef}{\stackrel{\text{def}}=}

\newcommand{\comment}[1]{}

\newcommand{\fun}[3]{\ensuremath{#1\colon #2 \to #3}}

\theoremstyle{plain}
\newtheorem*{rep@theorem}{\rep@title}
\newcommand{\newreptheorem}[2]{%
\newenvironment{rep#1}[1]{%
\def\rep@title{#2 \ref{##1}}%
\begin{rep@theorem}}%
{\end{rep@theorem}}}

\newreptheorem{theorem}{Theorem}
\newreptheorem{lemma}{Lemma}
\newreptheorem{proposition}{Proposition}
\newreptheorem{conjecture}{Conjecture}
\newreptheorem{claim}{Claim}

\title{\msonab is undecidable}
\author{
  \IEEEauthorblockN{Miko{\l}aj Boja{\'n}czyk\qquad Edon Kelmendi\qquad Micha{\l} Skrzypczak}
  
  \IEEEauthorblockA{University of Warsaw}
}

\begin{document}
 \IEEEoverridecommandlockouts
  \IEEEpubid{\makebox[\columnwidth]{978-1-7281-3608-0/19/\$31.00~
  \copyright2019 IEEE \hfill} \hspace{\columnsep}\makebox[\columnwidth]{ }}
\maketitle

\begin{abstract}
This paper is about an~extension of monadic second\=/order logic over the full binary tree, which has a~quantifier saying ``almost surely a branch $\pi \in \set{0,1}^\omega$ satisfies a~formula $\varphi(\pi)$''.  This logic was introduced by Michalewski and Mio; we call it 
\msonab following notation of  Shelah and Lehmann. The logic \msonab subsumes many qualitative probabilistic formalisms, including qualitative probabilistic {\sc ctl}, probabilistic {\sc ltl}, or parity tree automata with probabilistic acceptance conditions. 
We show that it   is undecidable to check if  a given~sentence of \msonab is true in the full binary tree\footnote{Independently and in parallel another proof of this result was given employing different techniques in \cite{berthon19_monad_secon_order_logic_with}.}.

\end{abstract}

\section{Introduction}
\label{sec:introduction}

Probability has been present in the theory of verification since the very beginning. An early example~\cite{vardi1985automatic,vardi1986automata} is the following question: given an {\sc ltl}  formula and a~Markov chain,  decide if almost all (in the sense of measure) runs of the system satisfy the formula. Another early example~\cite{hart1986probabilistic} is: given a~formula of probabilistic {\sc ctl}, decide if there is some Markov chain where the formula is true (the complexity of the problem is settled in~\cite{brazdil2008satisfiability}). The same question for the more general logic {\sc ctl}$^\ast$ is answered in \cite[Theorem~1 and~2, and Section~15]{lehman_time_and_chance}. Other variants of these logics have been considered in \cite{hansson1994logic,baier1998model}. More recent work tries to  synthesize controllers for probabilistic systems, see e.g.~\cite[Theorem 15]{berthon2017threshold}.


Is there a~master theorem, which unifies all decidability results about probabilistic logics? An inspiration for such a~master theorem would be Rabin's famous result~\cite{rabin_s2s} about decidability of monadic second\=/order logic over infinite trees.  Rabin's theorem  immediately gives most  decidability results (if not  the optimal complexities)  about temporal logics, including satisfiability questions for (non\=/probabilistic) logics like {\sc ltl}, {\sc ctl}$^*$ and the modal $\mu$\=/calculus.  Maybe there is a~probabilistic extension of Rabin's theorem, which  does the same for probabilistic logics?

Quite surprisingly, the  question about a~probabilistic version of Rabin's theorem has only been asked recently, by Michalewski and Mio~\cite{michalewski2016measure}. It is rather easy to see that any decidable version of {\sc mso} must be  qualitative rather than quantitative (\ie~probabilities can be compared to $0$ and $1$, but not to other numbers), since otherwise one could express  problems like ``does a~given probabilistic automaton accept some word with probability at least $0.5$'', which are known to be undecidable~\cite{paz1971introduction}, see also~\cite{gimbert2010probabilistic}. Even when probabilities are qualitative, one has to be careful to avoid undecidability. For example, the following problem is undecidable~\cite[Theorem~7.2]{baier2012probabilistic}: given a~B\"uchi automaton, decide if there is some $\omega$\=/word that is accepted with a~non\=/zero probability (assuming that runs of the automaton are chosen at random, flipping a~coin for each transition). This immediately implies~\cite[Theorem~1]{michalewski2016measure} undecidability for a~natural probabilistic extension of {\sc mso}, which has a~quantifier of the form ``there is a~non\=/zero probability of picking a~set $X$ of positions that satisfies $\varphi(X)$'', both for infinite words and infinite trees. 

Michalewski and Mio propose a~different probabilistic extension of {\sc mso}, which does not admit any straightforward reductions from known undecidable problems, like the ones for probabilistic B\"uchi automata mentioned above. Their idea---which only makes sense for trees and not words---is to extend {\sc mso} over the infinite binary tree by a~quantifier which says that a~property $\varphi(\pi)$ of branches is true almost surely, assuming the coin-flipping measure on infinite branches in the complete binary tree. The logic proposed by Michalewski and Mio is obtained from Rabin's {\sc mso} by adding the probabilistic quantifier for branches. We write \msonab for this logic\footnote{In~\cite{michalewski2016measure} the quantifier is denoted by $\forall^{=1}_\pi$, but in this paper we denote it by $\nabla$, following the notation used by Shelah and Lehmann in~\cite{lehman_time_and_chance}.}. As explained in~\cite{michalewski2016measure},  \msonab directly expresses qualitative problems like: model checking Markov chains for {\sc ltl} objectives, their generalisations such as~2{\small $\oldfrac{1}{2}$}
  player games with $\omega$\=/regular objectives, or emptiness for various automata models with probability  including the qualitative tree languages from~\cite{carayol2014randomization}. These results naturally lead to the question~\cite[Problem~1]{michalewski2016measure}: is the logic \msonab decidable?

A~positive result about \msonab was proved in~\cite{bojanczyk2016thin,bojanczyk2017emptiness}: the weak fragment of \msonab is decidable. In the weak fragment, the set quantifiers $\forall X$ and $\exists X$ of {\sc mso} range only over finite sets\footnote{Actually, the papers prove decidability for a~stronger logic, where set quantifiers range over ``thin'' sets, which are a~common generalisation of finite sets and infinite branches.}. The decidability proof uses automata: for every formula of the weak fragment there is an~equivalent automaton of a~suitable kind~\cite[Theorem~8]{bojanczyk2016thin}, and emptiness for these automata is decidable~\cite[Theorem~3]{bojanczyk2017emptiness}.  Combining these results, one obtains decidable satisfiability\footnote{For weak logics the satisfiability problem ``is a~given formula true in some infinite labelled binary tree'' is in general more difficult than the model checking problem ``is a~given formula true in the unlabelled binary tree''. For general {\sc mso}, this difference disappears, as set quantification can be used to guess labellings.} for the weak fragment of \msonab. The weak fragment of \msonab is still powerful enough to subsume problems like satisfiability for qualitative probabilistic {\sc ctl}${^*}$. Nevertheless, the decidability of the full logic \msonab remained open.

This paper proves that the full logic \msonab is undecidable, \ie~it is undecidable if a~sentence of the logic is true in the full binary tree, thus answering~\cite[Problem~1]{michalewski2016measure}. Independently and in parallel another proof of this result is given in \cite{berthon19_monad_secon_order_logic_with}, by proving that the emptiness problem of qualitative universal parity tree automata is undecidable. 




\section{The logic}
\label{sec:the-logic}
In this section we describe the logic \msonab. 

Our logic is an extension of Rabin's \mso over the full binary tree, so we begin by describing that. 
We write $\two$ for the set $\set{0,1}$. The full binary tree is the tree where nodes are identified with $\two^\ast$, finite words over the alphabet $\set{0,1}$. The ancestor order $\leq$ is the prefix relation. We write $|x| \in \nat$ for the length of a bit sequence $x \in \two^*$.  To express properties of the  full binary tree, we use monadic second-order logic (\mso). This logic which has two types of variables
\begin{align*}
  \underbrace{X,Y,Z,\ldots}_{\text{sets of nodes}} \qquad \underbrace{x,y,z,\ldots}_{\text{nodes}}
\end{align*}
which can be quantified  existentially and universally. To compare nodes and sets of nodes we use predicates
\begin{align*}
  x \in X \qquad \underbrace{x \le y}_{\text{ancestor}} \qquad \underbrace{x=y0}_{\text{left child}} \qquad \underbrace{x=y1}_{\text{right child}}.
\end{align*}
By 
Rabin's Theorem,  there is an algorithm which inputs a sentence of \mso, and says if the sentence is true in the full binary tree, see~\cite{thomas_languages} for a survey of the topic. 

The idea behind \msonab is to extend \mso with probabilistic quantification over branches\footnote{There is an alternative way of  adding probability to \mso, namely by having a quantifier which says that $\varphi(X)$ is true almost surely, assuming that  the set of nodes $X$ is chosen uniformly at random. This logic is already known to be undecidable~\cite[Theorem 1]{michalewski2016measure}, even for $\omega$-words, thanks to a straightforward reduction from emptiness for probabilistic B\"uchi automata with an almost sure acceptance condition~\cite{baier2012probabilistic}.}. A \emph{branch} is defined to be an element of $\two^\omega$.  Probability for sets of branches is measured using the \emph{coin\=/tossing} measure on $\two^{\omega}$,  which is the unique complete probabilistic  measure $\pr$ that satisfies
\[
  \pr\big[x\cdot\two^{\omega}\big]=2^{-|x|},
\]
for all $x\in\two^\ast$. The logic \msonab extends \mso by adding a new type of variable
\begin{align*}
  \underbrace{\pi,\sigma,\tau,\ldots}_{\text{branches}} 
\end{align*}
along with a membership test $x \in \pi$ (for membership tests, a branch is identified with the set of nodes that are its finite prefixes). To bind branches, the logic \msonab has a probabilistic quantifier
\begin{align*}
  \nabla \pi.\ \phi(\pi),
\end{align*}
which says that  there exists a set $R\subseteq \two^{\omega}$, such that $R$ has defined measure equal to 1, and every branch in $R$ satisfies $\phi$.  Intuitively, it means that $\phi(\pi)$ holds for a~\emph{randomly chosen} branch. This completes the definition of \msonab. 

We now give some examples that  illustrate the expressive power of \msonab.

\begin{example}\label{ex:fat-cantor}  This example is from~\cite[Section 3]{bojanczyk2016thin}.
    Consider the formula
\begin{align*}
    \exists X\ \begin{cases}
        \underbrace{\forall x.\ \exists y.\ (y \ge x \land y \in X)}_{\text{every node has a descendant in $X$}}\\
        \underbrace{\neg \nabla \pi.\ (\exists x. \ x \in \pi \wedge x \in X)}_{\text{with positive probability, $\pi$ avoids $X$}}
    \end{cases}
\end{align*}
This sentence is true. To see why, consider 
\begin{align*}
    X = \bigcup_{n \ge 2} X_n \qquad \text{where }X_n = \set{x0^{n} :  x \in \two^n}.
\end{align*}
Every node $x$ in the full binary tree has a descendant in $X$, namely $x0^{|x|}$. The probability of a branch visiting $X_n$ is   $1/2^n$, and therefore the probability of visiting $X$ is at most 
\begin{align*}
    \frac 1 2 = \sum_{n \ge 2} \frac 1 {2^n}.
\end{align*}
(In fact, the probability of visiting $X$ is smaller, because the events of visiting $X_n$ are not independent.)
It follows that the probability of avoiding $X$ is positive, and therefore $X$ makes the formula true.  One can show that  there is no set $X$ which makes the formula true and which is regular when seen as a language $X \subseteq \two^*$. This implies that  the family of  sets $X$ which make the formula true cannot be defined in \mso. 
\end{example}

The above example shows that  for formulas with free set  variables -- which can be seen as describing languages of labelled trees -- the logic \msonab is strictly more expressive than \mso.

\begin{example}\label{ex:carayol} Following~\cite{carayol2014randomization}, consider a nondeterministic parity automaton on infinite trees, where a run is considered accepting if the parity condition is satisfied almost surely. The existence of an accepting run can be easily expressed in the \msonab, by guessing a labelling of the tree with states and then checking the acceptance condition using the quantifier $\nabla$. The same idea works for more general acceptance conditions, e.g.~a conjunction of two acceptance conditions: an almost surely parity condition, and a usual (all paths) parity condition. Such automata are considered in~\cite{bojanczyk2016thin,bojanczyk2017emptiness,berthon2017threshold}.
\end{example}

\begin{example}\label{ex:pctl}
    Consider the following variant of qualitative probabilistic {\sc ctl}. This logic is used to define properties of labelled trees $t : \two^* \to \Sigma$. The atomic formulas check the label of the root, and Boolean combinations are allowed. There is a probabilistic version of the until operator: if $\varphi_1,\varphi_2$ are already defined formulas then also
    \begin{align*}
         \nabla (\varphi_1 \mathsf U \varphi_2),
    \end{align*}
    is a formula, which is true in a tree if almost surely a branch $\pi$ has the property that for some $y \in \pi$, the subtree of $y$ satisfies $\varphi_2$,  and for all $x < y$, the subtree of $x$   satisfies $\varphi_1$. For every formula $\varphi$ of this logic, one can easily write a sentence of \msonab that is true if and only if $\varphi$ is true in some labelled tree. The same kind of translation would work for many generalisations of the logic, e.g.~one could add an operator that checks if all (not almost all) paths satisfy a given property, or parity counting, etc.
\end{example}

The formulas in Examples~\ref{ex:fat-cantor}, \ref{ex:carayol} and~\ref{ex:pctl} are all of the form 
\begin{align}\label{eq:sigma1}
    \exists X_1 \ldots \exists X_n.\  \varphi(X_1,\ldots,X_n)
\end{align}
where $\varphi$ uses only $\nabla$, quantification over finite sets of nodes, and  (non-probabilistic) quantification over branches. By~\cite{bojanczyk2016thin,bojanczyk2017emptiness}, the truth of such sentences is  decidable\footnote{In the paper~\cite{berthon19_monad_secon_order_logic_with}, which gives an alternative proof of the main result in this paper, it is shown that universality is undecidable for the tree automata described in Example~\ref{ex:carayol}. It follows that the theory of \msonab is undecidable even after prepending universal set quantifiers in~\eqref{eq:sigma1}. Our undecidability proof uses formulas with a more complex quantifier structure.}. The purpose of this paper is to prove that, if we allow formulas that are more complicated  than~\eqref{eq:sigma1}, then the  logic becomes undecidable.


\section{Undecidability}
\label{sec:undecidability}
The main result of this paper is undecidability of the logic \msonab, as stated in the following theorem.

\begin{theorem}
  \label{thm:main} There is no algorithm which decides whether or not a given sentence of \msonab is true in the full binary tree.
\end{theorem}

The main ingredient in the undecidability proof is showing that \msonab can express a certain asymptotic counting property. Once the counting property has been defined, a routine encoding of Minsky machines can be used to establish undecidability. We now describe  this asymptotic counting property.

Define an~\emph{interval} to be a~finite path in the complete binary tree, \ie~a~set of the form
\begin{align*}
 \set{z : x \le z \le y} \qquad \text{for some $x,y \in \two^*, x < y$}.
\end{align*}
The nodes $x$ and $y$ are called the \emph{source} and \emph{target} of the interval, respectively. The \emph{interior}
of the interval $[x,y]$ is the set
\begin{align*}
  \inter([x,y])\eqdef\set{z : x < z < y}. 
\end{align*}
The \emph{length} of an interval is the cardinality of its interior. 

Let $\ifi$ be a family of intervals. If all intervals in $\ifi$ are  pairwise disjoint, then the family  is uniquely determined by the sets 
\begin{align*}
    \source \ifi, \target \ifi \subseteq \two^*
\end{align*}
of its sources and targets. We  only consider families of intervals that are pairwise disjoint, and therefore from now on, when we say \emph{family of intervals}, we mean a family of pairwise disjoint intervals. We write $\ifi, \ifj, \ifk$ for such families.


For a family of intervals $\ifi$ and a node $x$ that is the source of some interval $\ifi$, we write~$\ifi(x)$
for the length of the corresponding interval (which is unique by assumption that all intervals are pairwise disjoint). If $\pi$ is a branch, then we write $\ifi(\pi)$ for the sequence 
\begin{align*}
    \ifi(x_1),\ifi(x_2),\ldots
\end{align*}
where $x_1,x_2,\ldots$ are all of the sources of $\ifi$ that appear in $\pi$, ordered by increasing depth. See Figure~\ref{fig:seq}.
\begin{figure}[h]
  \includegraphics[width=0.48\textwidth]{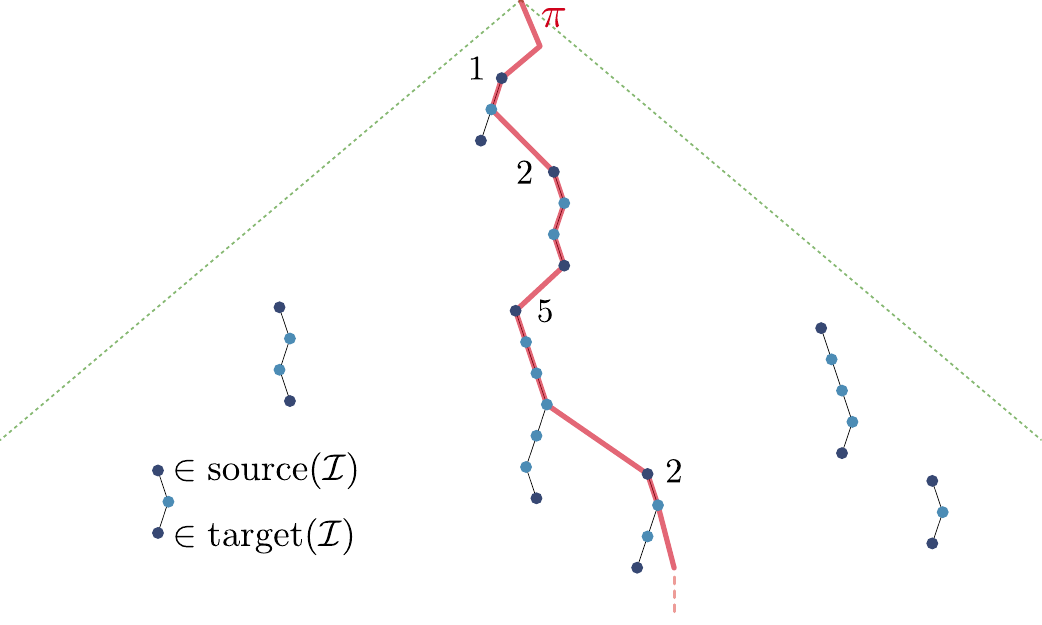}
  \caption{ Here we have the sequence $\ifi(\pi)=1,2,5,2,\ldots$. }
  \label{fig:seq}
\end{figure}

The sequence $\ifi(\pi)$ is a sequence of natural numbers, whose length may be  finite or infinite.
  We say that a sequence of natural numbers is \emph{eventually constant} if it has infinite length, and it has the same number on all but finitely many positions. Here is an example:
  \begin{align*}
      1,2,5,2,1,\overbrace{2,2,2,2,2,\ldots}^{\text{only $2$ }}.
  \end{align*}
  If $\ifi$ is a family of pairwise disjoint intervals, then we write
  \[\pr\big[\text{$\ifi$ is eventually constant}\big]\]
  for the probability of choosing a branch $\pi$ such that $\ifi(\pi)$ is eventually constant. The main technical result of this paper is that \msonab can express that this probability is 1. The family $\ifi$ is represented by its sources and targets.

\newcommand{\thmConstDef}{
    There is a formula $\varphi(X,Y)$ of \msonab which is true if and only if
    \[\pr\big[\text{$\ifi$ is eventually constant}\big]=1\]
for some\footnote{The family of intervals $\ifi$ is unique, if it exists.} family of intervals $\ifi$ where 
    \begin{align*}
        X = \source \ifi \qquad Y = \target \ifi.
    \end{align*}
    
}

\begin{theorem}
\label{thm:const-def}
\thmConstDef
\end{theorem}

Once we have proved the above lemma, undecidability of the logic follows by a routine reduction from the halting problem for Minksy machines. The general idea is to write the computation of the Minsky machine, repeated infinitely often, on each branch of the tree, and to use  eventually constant sequences to check if the counter values in consecutive configurations are consistent. This reduction is discussed in Section~\ref{sec:two-counter}.  The remaining part of the paper is devoted to proving Theorem~\ref{thm:const-def}.

Note how the property 
\[\pr\big[\text{$\ifi$ is eventually constant}\big]=1\]
is asymptotic in two ways: (a)~it allows sequences that are not eventually constant on a set of branches with zero probability, and (b)~on each branch there can be a finite delay before the constant tail starts.


\section{Boundedness properties}
The proof of Theorem~\ref{thm:const-def} builds on  ideas developed in the undecidability proofs from~\cite{bojanczyk_msou_final,u1u2} for the logic \msou, which is  quantitative extension of {\sc mso} that talks about boundedness. In this section, we   establish a connection with \msou, by showing that  \msonab can express various  boundedness properties for families of intervals. In the next section, we build on this connection, and known results about \msou, to express the language of eventually constant sequences in Theorem~\ref{thm:const-def}.

For a family $\ifi$ of pairwise disjoint intervals, let
\begin{align}\label{eq:liminf}
  \pr\big[\liminf \ifi < \infty \big]
\end{align}
be the probability of choosing a branch $\pi$ such that 
\begin{align*}
    \liminf \ifi(\pi) < \infty.
\end{align*}
The measured event is that the $\liminf$ is both defined (i.e.~$\ifi$ is visited infinitely often) and finite. 
In other words~\eqref{eq:liminf}  is the probability of choosing a branch such that $\ifi(\pi)$ contains some natural number infinitely often. The following lemma shows that \msonab can express positive probability of~\eqref{eq:liminf}. When we say that a formula of \msonab expresses a property of a family of intervals $\ifi$, we assume that $\ifi$ is given by two sets, representing its sources and targets,  as in Theorem~\ref{thm:const-def}.

Roughly, the main observation is as follows. If the intervals get progressively longer, then the probability of a branch visiting targets infinitely often drops to zero. Otherwise, if the intervals have bounded length, then almost every branch that visits sources infinitely often must also visit targets infinitely often. This phenomenon can be observed in Example~\ref{ex:fat-cantor}. In this example every node has a descendant in $X$, but in order to avoid these descendants with positive probability, they need to be progressively more and more distant. 

\begin{lemma}
    \label{lem:liminf}  \msonab can express
    \[\pr\big[\liminf \ifi < \infty \big] > 0.\]
  \end{lemma}
  \begin{proof}
    We show that the property in the statement of the lemma is equivalent to the following property, which is
    definable in \msonab (see Appendix~\ref{app:definability}).
      \begin{itemize}
          \item[$(\ast)$] there exists $\ifi'\subseteq \ifi$ such that 
          \begin{align*}
            \pr\big[\underbrace{ {\ifi'}\ \io}_{\substack{\text{a branch visits}\\ \text{sources of $\ifi'$}\\ \text{infinitely often}}} \big]>0
          \end{align*}
          and  all $\ifj\subseteq \ifi'$ satisfy
          \[\pr\big[ \ifj\ \io\ \Rightarrow\ \target \ifj\ \io\big]=1.\] 
      \end{itemize}
    
    \noindent{($\Rightarrow$)}
    We first show that the property in the statement of the lemma implies $(\ast)$. For $n \in \nat$, define $\ifi_n$ to be the intervals in $\ifi$ that have length  exactly $n$.  The event in the statement of the lemma says that with positive probability, there is some $n$ such that a branch passes through $\ifi_n$ infinitely often. By countable additivity of measures,  for some $n$ there is positive probability of seeing sources from $\ifi_n$ infinitely often. Define $\ifi' = \ifi_n$. To establish $(\ast)$, we prove the  following claim.
    \begin{claim}
      \label{lem:ob1}
      If all intervals in $\ifj$ have length $n$ then
      \[\pr\big[ \ifj\ \io\ \Rightarrow\ \target \ifj\ \io\big]=1.\]
    \end{claim}
    \begin{proof}

      Consider the complement of the event in the claim, that is:
      \begin{align*}
        \big[ \ifj\ \io\ \wedge \underbrace{\target \ifj\ \fo}_{\substack{\text{a branch visits}\\ \text{targets of $\ifj$}\\ \text{only finitely often}}}\big].
      \end{align*}
      It is equal to:
      \begin{align*}
        \bigcup_{x\in\two^{\ast}}\big[\ifj\ \io\  \wedge\ A(x)\big], 
      \end{align*}
      where by $A(x)$ we denote the event of a branch passing through $x$ and not visiting \emph{any} target of $\ifj$ after $x$. If $x_0$ is a source of $\ifj$, conditional on visiting $x_0$, the probability of the event $A(x_0)$ is at most $1-\frac{1}{2^{n+1}}$; to avoid every target below $x_0$, we have to avoid first the target corresponding to $x_0$ which is at distance $n+1$ since every interval in $\ifj$ has length $n$. In other words, when going down the tree from $x_0$ whenever we visit a source, the relative probability of further avoiding targets is at most $1-\frac{1}{2^{n+1}}$, which means that
      \begin{align*}
        \pr\big[\ifj\ \io\ \wedge\ A(x_0)\big]\le \lim_{k\to\infty}(1-\frac{1}{2^{n+1}})^k=0.
      \end{align*}
      This proves that the complement of the event in the claim has probability zero\footnote{A more direct (but abstract) proof of this claim can be given using L\'evy's zero-one law.}.
    \end{proof}
    
    \noindent{($\Leftarrow$)} We now show that $(\ast)$ implies the property in the statement of the lemma. Let then $\ifi'$ be as $(\ast)$. Since the property in the statement of the lemma is closed under adding intervals to a family, it is enough to show 
    \[\pr\big[\liminf \ifi' < \infty \big] > 0.\]
    We will show a stronger property, namely
    \begin{align}\label{eq:finite-limsup}
      \pr\big[\limsup \ifi' < \infty \big] > 0.
    \end{align}
    An interval in $I \in \ifi'$ is called a \emph{record breaker} if it is strictly longer than  all intervals in $\ifi'$ with sources that are ancestors of the source of~$I$. 
    \begin{claim}
    Almost surely, the sources of record breakers are visited finitely often.  
    \end{claim}
    \begin{proof}
      Define $A_n$ to be the branches which see the target of some record breaker after having already seen at least $n$ sources of record breakers. By definition, on each branch, the $n$-th record breaker  has length at least $n$, and therefore the probability of seeing its target  is at most $\frac{1}{2^{n+1}}$. It follows that the probability of $A_n$ is at most 
      \begin{align*}
        \frac 1 {2^{n}} = \frac 1 {2^{n+1}} + \frac 1 {2^{n+2}} + \cdots.
      \end{align*}   Branches that visit infinitely many targets of record breakers belong to all sets $A_n$, and therefore they have probability zero.

      We have thus established that almost surely targets of record breakers are seen finitely often. If we set $\ifj$ to be the record breakers, then we know by $(\ast)$ that almost surely sources of record breakers are seen finitely  often, thus establishing the claim.  
    \end{proof}
     A branch $\pi$ sees record breakers infinitely often  if and only if the sequence $\ifi'(\pi)$ is has infinite $\limsup$. Therefore, it follows from the claim that almost surely the sequence $\ifi'(\pi)$ has finite length, or it is infinite but has finite $\limsup$.  Since there is positive probability of visiting $\ifi'$ infinitely often, we get~\eqref{eq:finite-limsup}.
        \end{proof}

Building on the above lemma, we now show how \msonab can characterise branches $\pi$ where $\ifi(\pi)$ is unbounded. 
\begin{definition}
\label{def:unbounded}
A set of nodes $X$ is called a \emph{characteristic} for  a family of intervals $\ifi$  if
\begin{equation}
  \label{eq:char-def}
  \pr\big[X\ \io \iff (\limsup \ifi = \infty)\big]=1.
\end{equation}
\end{definition}
Recall the notion of record breakers that was used in the proof of Lemma~\ref{lem:liminf}. It is not hard to see that the record breakers are a characteristic, and therefore every family of intervals admits at least one characteristic.  The following lemma shows that being a characteristic can be described in \msonab (we assume, as usual, that a family of intervals is given by its sources and targets). 

\begin{lemma}
  \label{lem:char} There is a formula of 
  \msonab which says that $X$ is a characteristic of $\ifi$. 
\end{lemma}
\begin{proof}
  We say that   $Y$ is  a \emph{semi-characteristic} of $\ifi$ if
  \begin{align}\label{eq:semi-char}
  \pr\big[Y \ \io \ \Rightarrow \ (\limsup \ifi = \infty)\big] = 1.
  \end{align}
  It is not hard to see that $X$ is a characteristic of $\ifi$ if and only if every semi-characteristic satisfies 
  \begin{align*}
    \pr\big[Y \ \io \ \Rightarrow \ X \ \io\big] = 1.
  \end{align*}
  Therefore, to prove the lemma, it is enough to define semi-characteristics in \msonab.

  We claim that~\eqref{eq:semi-char} is equivalent to 
       \begin{itemize}
       \item[$(\ast)$] There exists $\ifj\subseteq \ifi$ which is unbounded and
         \begin{align*}
           \pr\big[ Y\ \io\ \Rightarrow\ \ifj\ \io  \big]=1.
         \end{align*}
       \end{itemize}
  We say that a family $\ifj$ is \emph{unbounded} if
  \begin{align*}
    \pr\big[\ifj\ \io \ \Rightarrow\ (\limsup\ifj=\infty) \big]=1.
  \end{align*}
  Being unbounded is equivalent to saying:
  \begin{equation}
    \label{eq:unbounded formula}
    \begin{aligned}
      &\text{there exists $\ifj'\subseteq\ifj$ such that}\\
      &\qquad\pr\big[ \ifj\ \io\ \iff\ \ifj'\ \io\big]=1\text{ and}\\
      &\qquad\pr\big[\ifj'\ \io\ \Rightarrow\ (\liminf\ifj'=\infty)\big]=1.
    \end{aligned}
  \end{equation}
  To see this take $\ifj'$ to be the record breakers for the forward implication; the converse is immediate. Further, the condition \eqref{eq:unbounded formula} (and therefore also $(\ast)$) is definable in \msonab, since the second conjunct is the complement of the property from Lemma~\ref{lem:liminf}.

  The implication~$(\ast) \Rightarrow $\eqref{eq:semi-char} is trivial. For the converse implication, we take $\ifj$ to be the record breakers. Then the family $\ifj$ is unbounded and we have $(\ast)$ because for every branch $\pi$, $\limsup\ifi(\pi)=\infty$ if and only if $\ifj$ appears infinitely often in $\pi$.

\end{proof}


\section{Eventually constant intervals}
\label{sec:econst}

We call elements of $\nat^\omega$ \emph{number sequences}. They are denoted by $f,g,h$. In the previous section, we have essentially encoded number sequences on branches using intervals, and demonstrated that the probabilistic quantifier can be used to say that the encoded number sequences are bounded (they have finite $\limsup$) almost surely. How is boundedness useful for expressing the eventually constant language in Theorem~\ref{thm:const-def}? To answer this question, we first need to define asymptotic mixes. The ideas are borrowed from the proof of undecidability of \msou in~\cite{bojanczyk_msou_final}.

If $X=\{x_0<x_1<\ldots\}\subseteq\nat$ then by $f\restriction_X$ we denote the subsequence of $f$ taking only positions from $X$, \ie~$f\restriction_X=(f(x_0), f(x_1), \ldots)\in\nat^\ast\cup \nat^\omega$.

\begin{definition}[Asymptotic equivalence]
  Given $f,g\in\nat^\omega$, we say that $f$ is \emph{asymptotically equivalent} to $g$, denoted $f\sim g$, if $f$ and
  $g$ are bounded on the same sets of positions, \ie~for all $X\subseteq \nat$, either both $f\restriction_X$ and
  $g\restriction_X$ are bounded or both are unbounded.
  If $f$ is not asymptotically equivalent to $g$ we write~$f\not\sim\ g$. 
\end{definition}

A~\emph{vector sequence} is an element of $(\nat^{{+}})^\omega$, \eg:
\begin{align*}
  (4,7,6)\ (2,3)\ (10)\ (1,1,1) \cdots.
\end{align*}
We denote vector sequences by~$\vsf,\vsg,\vsh$. We say that a~number sequence $f\in\nat^\omega$ is an~\emph{extraction}
of $\vsf$ (denoted $f\in\vsf$) if for each $n\in\nat$ the number $f(n)$ is a~component of $\vsf(n)$ (written simply
$f(n)\in\vsf(n)$).

\begin{definition}[Asymptotic mix]
\label{def:assmix}
  Given two vector sequences $\vsf$, $\vsg$ we say that $\vsf$ is an~\emph{asymptotic mix} of $\vsg$ if for all $f\in\vsf$
  there exists $g\in\vsg$ such that $f\sim g$.
\end{definition}

A~vector sequence $\vsf$ has \emph{dimension} $d$ if every vector in it has dimension $d$. Notice that each vector of a~vector sequence must be non\=/empty and therefore, $d\geq 1$ always.
The following lemma (that we
state without a~proof) makes a~crucial connection between the dimension and asymptotic mixes, the latter being
a~property of boundedness of the components of vector sequences.

\begin{lemma}[\cite{bojanczyk_msou_final} Lemma~2.1]
  \label{lem:msoulemma}
  Let $d\in \nat$, $d>0$. There exists a~vector sequence of dimension $d$ which is not an~asymptotic mix of any vector sequence of dimension $d-1$ (nor any smaller dimension).
\end{lemma}

We will encode vector sequences with two families of intervals $\ifj$ and $\ifi$, by \emph{wrapping} the former over the latter. The lengths of $\ifj$ will encode the dimensions, and those of $\ifi$ will encode the components. We want to express that $\ifj$ is eventually constant. The rough idea is as follows. If $\ifj$ is not eventually constant then it must alternate between two lengths (we can ask for it to be bounded), say $5$ and $3$. We then check whether this is the case by employing Lemma~\ref{lem:msoulemma}.

But it is not yet clear how we are to express asymptotic equivalence and mixes in \msonab, so we do this first in the next two technical subsections. 

\subsection{Asymptotic equivalence}
Consider \mso on infinite words for a moment. Suppose that we encode two number sequences with families of intervals
$\ifi_1$, $\ifi_2$. {\em A~priori} it is not possible to express $\ifi_1\sim\ifi_2$ in the logic\footnote{Even if we are
  allowed to speak about boundedness.}, unless we impose some restriction, such that there is some \mso definable function
that given the $n$th interval of $\ifi_1$ outputs the position of the $n$th interval of $\ifi_2$. The simplest way of
having this is to require that the intervals in $\ifi_1$ and $\ifi_2$ are alternating:
\vspace{-12pt}
\begin{center}
  \includegraphics[width=0.5\textwidth]{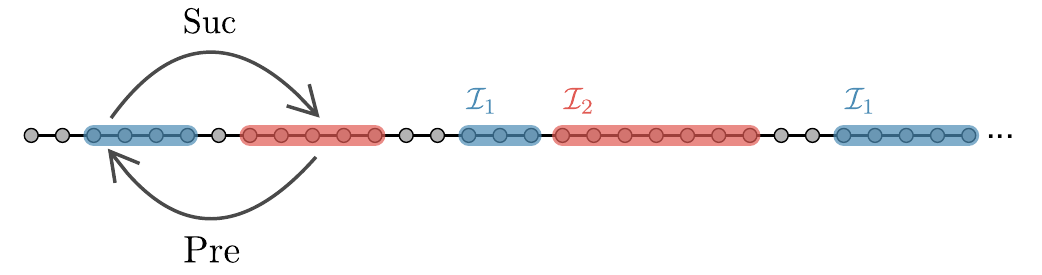}
\end{center}
\vspace{-8pt}
If $\ifi_1$, $\ifi_2$ are arranged in such a way, the functions $\pre$ and $\suc$ are \mso definable (the first neighbour to the left, or right respectively) and hence we are able to quantify over subsequences which enables us to express asymptotic equivalence in the logic.

For trees we have the following definitions. 

We call two families of intervals $\ifi_1$, $\ifi_2$ \emph{isolated} if $\bigcup\ifi_1\cap \bigcup\ifi_2=\emptyset$, \ie~there is no node that belongs both to an interval in $\ifi_1$ and an interval in $\ifi_2$
\begin{definition}[Precedes]
  \label{def:precedes}
  Let $\ifi_1$, $\ifi_2$ be isolated families of intervals. We say that $\ifi_1$ \emph{precedes}~$\ifi_2$ if for all
  $x'\in\source{\ifi_2}$ there exists $x\in\source{\ifi_1}$ such that $x<x'$ and there is no node strictly between $x$ and $x'$
  that is a source of $\ifi_1$ or $\ifi_2$.
\end{definition}

The fact that $\ifi_1$ precedes $\ifi_2$ induces a~function $\fun{\pre}{\source{\ifi_2}}{\source{\ifi_1}}$ that maps $x'\mapsto x$ as in the definition above. Additionally, for a~family $\ifi\subseteq\ifi_1$, we define:
\[
  \suc(\ifi)\eqdef\big\{[x',y']\in\ifi_2\st\pre(x')\in \source{\ifi}\big\}\subseteq \ifi_2,
\]
and dually, for $\ifi\subseteq\ifi_2$ we put
\[
  \pre(\ifi)\eqdef\big\{[x,y]\in\ifi_1\st \exists x'\in\source{\ifi}.\ \pre(x')=x\big\}.
\]
For the sake of readability we will use the functions $\pre$ and $\suc$ without additional parameters, assuming that the familiesv $\ifi_1$ and $\ifi_2$ are known from the context.
The picture on trees looks as follows:
\vspace{-12pt}
\begin{center}
\includegraphics[width=0.5\textwidth]{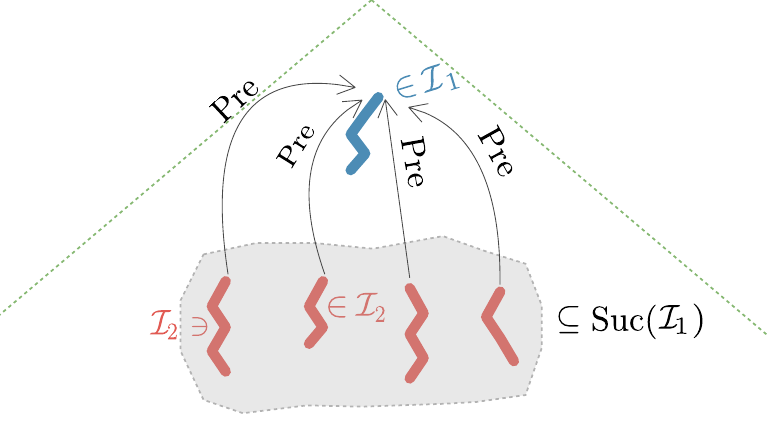}
\end{center}
\vspace{-10pt}
In a branch $\pi$, it might be the case that between consecutive intervals in $\ifi_2$, there are many sources of intervals from
$\ifi_1$, so the encoding of the two sequences is not alternating, hence the following definition.

\begin{definition}[Preceding subsequence]
  \label{def:precsubseq}
  Let $\ifi_1$, $\ifi_2$ be isolated families of intervals such that $\ifi_1$ precedes~$\ifi_2$. Assume that $\pi$ is a~branch where
  $\ifi_2$ appears infinitely often. By $\ifi_1^\pre(\pi)$ we denote the subsequence of $\ifi_1(\pi)$ that we get by applying $\ifi_1$
  only to the nodes $x$ for which there exists $x'\in\pi\cap\source{\ifi_2}$ such that~$\pre(x')=x$. 
\end{definition}

Notice that in the above definition we require $x'$ to belong to $\pi$, \ap we might have $\pre(x')=x$ for some $x'\in\source{\ifi_2}$ outside $\pi$ but for no such node in $\pi$ (in that case $\ifi_1(x)$ is not taken into $\ifi_1^\pre(\pi)$).
Observe additionally that if $\ifi_1$ precedes $\ifi_2$ and $\ifi_2$ appears infinitely often in a~branch $\pi$ then $\ifi_1^\pre(\pi)$ is a~number sequence (\ie~it is infinite). However, we are not claiming that $\ifi_1^\pre$ is a~family of intervals. 

Typically, on a branch $\pi$ where $\ifi_2$ appears infinitely often we have: a few intervals of $\ifi_1$ then one interval in $\ifi_2$ and so on. The sequence $\ifi_1^\pre(\pi)$ is taking into account only the intervals that immediately precede those of~$\ifi_2$. It looks as follows:

\begin{center}
  \includegraphics[width=0.5\textwidth]{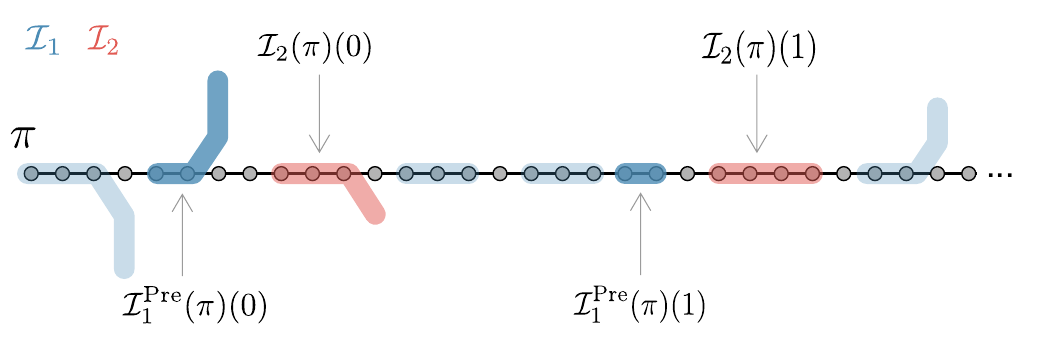}
\end{center}
\vspace{-12pt}
\begin{remark}
\label{rem:suc-correspondence}
Consider $\ifi_1$, $\ifi_2$ two isolated families of intervals such that $\ifi_1$ precedes $\ifi_2$. Let $\pi$ be a~branch on which $\ifi_2$ appears infinitely often. In that case the two sequences $\ifi_1^\pre(\pi)$ and $\ifi_2(\pi)$ are both defined. Let $x_k'\in\pi\cap\source{\ifi_2}$ be the $k$th source of an~interval in $\ifi_2$ on $\pi$ (it has $k-1$ strict ancestors in $\source{\ifi_2}$). Then, by the definitions of the respective sequences:
\begin{align*}
\ifi_2(\pi)(k)&= \ifi_2(x_k'),\\
\ifi_1^\pre(\pi)(k) &=\pre(\ifi_2)\big(\pre(x_k')\big).
\end{align*}
This means that the two number sequences are in a~sense \emph{synchronised} and the function $\pre$ maps between the corresponding sources.

In other words, number sequence encodings $\ifi_2$ and~$\ifi_1^\pre$ are alternating as in the case of infinite words, which facilitates quantifying over
their subsequences.
\end{remark}

As a consequence it is easier to express asymptotic equivalence between $\ifi_2(\pi)$ and $\ifi_1^\pre(\pi)$.
\newcommand{\lemcpre}[1]{
  Let $\ifi_1$, $\ifi_2$ be isolated families of intervals, such that $\ifi_1$ precedes $\ifi_2$. Then we can express in \msonab that:
  \begin{align*}
    \pr\big[\ifi_2\ \io\wedge \ifi_1^{\pre}\not\sim\ifi_2\big]>0.
  \end{align*}
}

\begin{lemma}
  \label{lem:c1pre}
  \lemcpre{}
\end{lemma}

The formula used to express the property in the lemma above utilizes the fact that $\pre$ and $\suc$ are \mso-definable and quantifies over subsets of $\pre(\ifi_2)$ and $\ifi_2$. The proof is in Appendix~\ref{app:lemc1pre}. 

\subsection{A characterization of asymptotic mixes}

Having built tools to express asymptotic equivalence, we now move on to asymptotic mixes. In this section we give the definition of \emph{separation} which is equivalent to asymptotic mixes.
\begin{remark}
  \label{rem:whyseparation}
  The reason why we give this equivalent definition of asymptotic mixes is that it will allow us in the sequel
  to partition certain sets of branches into countably many subsets (one for each bound $b$), for the purpose of then
  using the $\aleph_0$-additivity of the measure. Thereby allowing us to \emph{pull out} one existential quantifier. 
\end{remark}

For a~vector sequence $\vsf$ denote by $\min(\vsf)\in\vsf$ (respectively $\max(\vsf)\in\vsf$) the number sequences that pick the
minimal (respectively maximal) component of every vector. For a~number sequence $f\in\nat^\omega$ and $b\in\nat$ we write $f\leq b$ if
for all $n\in\nat$ we have $f(n)\leq b$.

\begin{definition}[Separation]
\label{def:separation}
  Let $\vsf$, $\vsg$ be two vector sequences and $b\in\nat$. We say that $b$ \emph{separates} $\vsf$ from~$\vsg$ if one of
  the following holds:
  \begin{itemize}
  \item $\exists X.\ \min(\vsf\restriction_X)\leq b$ and $\min(\vsg\restriction_X)$ is unbounded,
  \item $\exists X.\ \max(\vsg\restriction_X)\leq b$ and $\max(\vsf\restriction_X)$ is unbounded. 
  \end{itemize}
\end{definition}

\begin{lemma}
  \label{lem:separates}
  Let $\vsf$, $\vsg$ be two vector sequences. Then~$\vsf$ is not an asymptotic mix of $\vsg$ if and only if there exists
  $b\in\nat$ that separates $\vsf$ from $\vsg$. 
\end{lemma}

\begin{proof}
  We start with the forward implication. Given a~number sequence $f$ we define the \emph{best response} $g_f\in\vsg$ for $n\in\nat$ as
  \[
    g_f(n)=\argmin_{x\in \vsg(n)}|f(n)-x|. 
  \]
  So $g_f$ is the choice of components in $\vsg$ that minimize the distance to $f$.

  Since $\vsf$ is not an~asymptotic mix of $\vsg$, there exists $f\in\vsf$ such that for all $g\in\vsg$, $f\not\sim g$;
  in particular we have $f\not\sim g_f$. This means that there exists $X\subseteq \nat$ such that one of the following
  holds:
  \begin{itemize}
  \item $f\restriction_X$ is bounded and $g_f\restriction_X$ is unbounded,
  \item $g_f\restriction_X$ is bounded and $f\restriction_X$ is unbounded.
  \end{itemize}
  By the definition of $g_f$, in the first case $\min(\vsg\restriction_X)$ is unbounded while $\min(\vsf\restriction_X)$ is clearly bounded (by some $b\in\nat$). In the second case we have
  $\max(\vsg\restriction_X)\leq b$ for some $b$ while $\max(\vsf\restriction_X)$ is unbounded. Therefore, there exists $b\in\nat$ that separates $\vsf$
  from~$\vsg$. 
  
  For the backward implication, assume that $b$ separates $\vsf$ from $\vsg$. In the first
  case of Definition~\ref{def:separation} it suffices to construct $f\in\vsf$ by picking a~component smaller than $b$ if it exists, and an~arbitrary
  component otherwise. In the second case, we pick the maximal component.
\end{proof}

\subsection{Wrappings}
\label{sec:wrappings}

Let us now explain in more detail how vector sequences are encoded using families of intervals. Recall the definition of $\inter$ from page~\pageref{thm:main}.

\begin{definition}[Wrappings]
\label{def:wrapping}
Let $\ifi$, $\ifj$ be families of intervals. We say that $\ifj$ \emph{wraps} $\ifi$ if $\inter(\ifj)=\source\ifi$ and for each interval $[x,y]\in\ifj$ we have $\length([x,y])\geq 1$.
\end{definition}

Let $\ifi$, $\ifj$ be families of intervals such that $\ifj$ wraps $\ifi$ and take $[x,y]\in\ifj$. Then
$\inter([x,y])=\set{x_1,x_2,\ldots,x_{\ifj(x)}}$ such that $x<x_1<\cdots < x_{\ifj(x)}<y$ and $\ifj(x)\geq 1$. All the $x_i$s are sources
of some intervals in $\ifi$. Define:
\[
  \vec{\ifj}(\ifi,x)=\big(\ifi(x_1),\ifi(x_2),\ldots,\ifi(x_{\ifj(x)})\big). 
\]

Extend this definition to branches $\pi$ in such a~way that if $\ifj$ appears infinitely often in $\pi$ then $\vec{\ifj}(\ifi,\pi)$ is a~vector sequence:
if $\pi\cap\source\ifj=\{x_0<x_1<\ldots\}$ then $\vec{\ifj}(\ifi,\pi)(k)$ equals $\vec{\ifj}(\ifi,x_k)$.

In this way we can encode vector sequences using two families of intervals $\ifi$, $\ifj$. The lengths of intervals in the
outer layer $\ifj$ are the dimensions of the vectors, while the lengths of the intervals in $\ifi$ are the
components. We illustrate this in Figure~\ref{fig:wrapping}.
\begin{figure}[h]
  \includegraphics[width=0.5\textwidth]{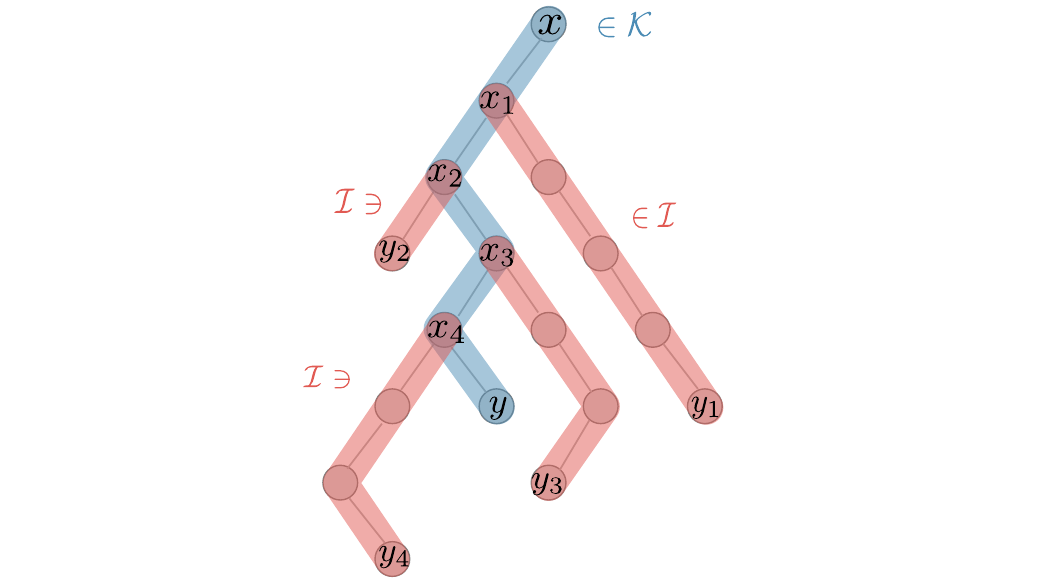}
  \caption{ In this partial tree the set $[x,y]$ is an~interval in~$\ifj$, and $[x_i,y_i]$ are intervals in $\ifi$,
    $1\leq i\leq 4$. We have $\ifj(x)=4$, $\ifi(x_1)=3$, $\ifi(x_2)=0$, $\ifi(x_3)=2$, and $\ifi(x_4)=2$. The vector
    that is encoded in $x$ is $\vec{\ifj}(\ifi, x)=(3,0,2,2)$.}
  \label{fig:wrapping}
\end{figure}


Just as in the case of asymptotic equivalence above, it is not possible to speak of whether a vector sequence that is encoded using families of intervals as above is an asymptotic mix of another, without imposing some structure. For this purpose, for number sequences we had the notion of a family preceding another in Definition~\ref{def:precedes}, for vector sequences we use a stronger condition. 

\begin{definition}[Tail-precedes]
\label{def:tail-prec}
Let $\ifj_1$, $\ifj_2$ be isolated families of intervals. We say that $\ifj_1$ \emph{tail\=/precedes} $\ifj_2$ if for all $x'\in\source{\ifj_2}$ there exists $y\in\target{\ifj_1}$ such that $y<x'$ and there is no node strictly between $y$ and $x'$ that is a source of $\ifj_1$ or $\ifj_2$.
\end{definition}

Note that tail\=/preceding is a~stronger property than preceding given in Definition~\ref{def:precedes}, therefore if
$\ifj_1$ tail\=/precedes $\ifj_2$, and $\ifj_2$ appears infinitely often in some branch $\pi$ then the sequences $\ifj_1^\pre(\pi)$ and
$\vec{\ifj}_1^\pre(\pi)$ are well\=/defined. This enables us to talk about asymptotic mixes and to apply Lemma~\ref{lem:c1pre}.

However, what guarantees that the relevant families of intervals are structured in such a way, \ie one tail-preceding the other. This is the subject of the next lemma, which essentially says that if $\ifj$ is bounded but not eventually constant then it is possible to find two subfamilies each of certain lengths such that one tail-precedes the other. 

\begin{lemma}
  \label{lem:d1d2}
  Let $\ifj$ be a~family of intervals such that 
  \begin{align*}
    \pr\Bigg[\wedge\begin{cases}
      \ifj\ \io\\
      \limsup\ifj <\infty\\
      \text{$\ifj$ is not eventually constant}
    \end{cases}\Bigg]>0.
  \end{align*}
  Then there exist two numbers $\ell_1>\ell_2\in\nat$ and isolated $\ifj_1, \ifj_2\subseteq \ifj$ such that:
  \begin{itemize}
  \item every interval in $\ifj_1$ has length $\ell_1$,
  \item every interval in $\ifj_2$ has length $\ell_2$,
  \item $\ifj_1$ tail-precedes $\ifj_2$, and 
  \item $\pr\big[\ifj_2\ \io\big]>0$.
  \end{itemize} 
\end{lemma}

\begin{proof}
  We have assumed that there is a~non\=/zero probability of picking a~branch $\pi$ such that $\ifj(\pi)$ is a~sequence that
  is infinite, bounded, and not eventually constant. This means that with a~positive probability there are two numbers that both appear infinitely often in the sequence $\ifj(\pi)$, \ie the set
  \begin{align*}
    \left\{\pi\st \begin{aligned}
      &\exists \ell_1>\ell_2\in\nat.\\
      &\text{$\ifj(\pi)$ contains infinitely often $\ell_1$ and $\ell_2$}
      \end{aligned}\right\},
   \end{align*}
  has non\=/zero probability
  Consequently, as there are countably many choices of $\ell_1>\ell_2\in\nat$, there exist two numbers $\ell_1>\ell_2\in\nat$ such that:
  \[
    \pr\big[\text{$\ifj$ contains infinitely often $\ell_1$ and $\ell_2$}\big]>0.
  \]
  Let $\ifi_1\subseteq\ifj$ (respectively $\ifi_2\subseteq\ifj$) be the intervals in $\ifj$ whose length is
  $\ell_1$ (respectively $\ell_2$). The probability that both $\ifi_1$ and $\ifi_2$ appear infinitely often is non\=/zero. This means that:
  \begin{align*}
    \pr\big[\source{\ifi_1}\ \io\wedge\source{\ifi_2}\ \io\big]>0. 
  \end{align*}
  From Claim~\ref{lem:ob1} we have:
  \begin{align*}
    \pr\big[\target{\ifi_1}\ \io \wedge \source{\ifi_2}\ \io\big]>0. 
  \end{align*}
  
  Now we prove that we can find subsets of $\ifi_1,\ifi_2$ for which the last two bullet points in the statement of the lemma hold. 
  \begin{claim}
  \label{claim:alternatingxy}
  Let $X$, $Y\subseteq \two^\ast$ be such that $\pr\big[X\ \io\wedge Y\ \io\big]>0$. Then there exist $X'\subseteq X$ and $Y'\subseteq Y$
  such that between any two nodes $x<y$ in $Y'$ there exists a~node $u\in\inter([x,y])$ that belongs to $X'$ and moreover $\pr[Y'\ \io\wedge X'\ \io]>0$. 
\end{claim}

\begin{proof}
  We construct for all $n>0$, sets $X_n\subseteq X$, $Y_n\subseteq Y$ and put $X'=\bigcup_{n>0} X_n$, $Y'=\bigcup_{n>0} Y_n$. For
  any node $y$ we say that $x\in X$ is an~$X$\=/successor of $y$ if $x>y$ and there is no node strictly between $x$ and
  $y$ that is in $X$. Similarly we define $Y$\=/successors.

  Let $Y_0=\set{\epsilon}$ where $\epsilon$ is the root node and define for all $n>0$:
  \begin{align*}
    X_{n}\eqdef&\bigcup_{y\in Y_{n-1}}\set{x\in X\st \text{$x$ is an~$X$\=/successor of $y$}},\\
    Y_{n}\eqdef&\bigcup_{x\in X_{n}}\ \;\set{y\in Y\st \text{$y$ is a~$Y$\=/successor of $x$}}. 
  \end{align*}

  We can easily observe that for $X'$, $Y'$ constructed this way we have that between every two nodes in $Y'$ there is always a~node in $X'$ (in fact, also symmetrically, the nodes in $X'$ are separated by nodes in $Y'$). Let $\pi$ be a~branch where both $X$ and $Y$ appear infinitely often. Then the first non\=/root node in this branch that belongs to $X$
  belongs to $X_1$, after which the first node that belongs to $Y$ belongs to $Y_1$, and so on. Consequently both $X'$
  and $Y'$ also appear infinitely often in $\pi$. Therefore, $\pr\big[Y'\ \io\wedge X'\ \io\big]>0$.
\end{proof}

  Set $X=\target{\ifi_1}$, $Y=\source{\ifi_2}$ and apply Claim~\ref{claim:alternatingxy} resulting in $X'\subseteq X$ and
  $Y'\subseteq Y$. We set $\ifj_1$ (respectively $\ifj_2$) to be the intervals whose targets are in $X'$ (respectively
  sources in $Y'$). The statement of the lemma now can be deduced from the properties of $X'$ and $Y'$.
\end{proof}

Given families of intervals $\ifi,\ifj$ such that the latter wraps the former, we are now able to sketch how to express in \msonab that $\ifj$ is eventually constant.



\newcommand{\proconstant}[1]{
  Let $\ifi$, $\ifj$ be two families of intervals such that $\ifj$ wraps $\ifi$ and we have:
  \begin{align*}
    \pr\Bigg[\begin{cases}
      \ifj\ \io &\Rightarrow \limsup\ifj <\infty, \text{ and}\\
      \ifi\ \io &\Rightarrow \liminf\ifi = \infty
      \end{cases}\Bigg]=1.
  \end{align*}
  Then the following property is definable in \msonab:
  \begin{align*}
    \pr\big[\ifj\ \io\ \wedge \text{$\ifj$ is not eventually constant}\big]>0.
  \end{align*}
}

\begin{proposition}
  \label{pro:constant}
  \proconstant{}
\end{proposition}

Let $\ifi$, $\ifj$ be such that $\ifj$ wraps $\ifi$. We say that $\ifi'\subseteq\ifi$ is an~\emph{extraction} of
$(\ifj,\ifi)$ if for all $[x,y]\in\ifj$ there is exactly one element of $\source{\ifi'}$ that belongs to $\inter([x,y])$.

We write $\ifi_1\leq \ifi_2$ if the sources of the two families of intervals coincide and the targets of $\ifi_1$ are ancestors of the targets of $\ifi_2$, \ie~for every interval $[x,y]\in\ifi_1$ there is an~interval $[x,y']\in \ifi_2$ such that $[x,y]\subseteq [x,y']$ (equivalently $y\leq y'$).

We claim that the statement in the proposition is equivalent to the following:
\begin{itemize}
  \item[$(\ast)$] there exist isolated  $\ifj_1,\ifj_2\subseteq\ifj$, where $\ifj_1$ tail\=/precedes $\ifj_2$, $\pr[\ifj_2\ \io]>0$, and if
    $\ifi_1,\ifi_2\subseteq \ifi$ are such that $\ifj_i$ wraps $\ifi_i$, $i\in\set{1,2}$ then:
    \begin{align}
      &\exists\ifi_1'\leq\ifi_1.\ \ \forall \ifi_2'\leq\ifi_2.\nonumber\\
      &\ \ \exists\ifk_1\subseteq\ifi_1'\text{ extraction of }(\ifj_1,\ifi_1').\nonumber\\
      &\ \ \forall\ifk_2\subseteq\ifi_2'\text{ extraction of }(\ifj_2,\ifi_2').\nonumber\\
      &\qquad\pr\big[\ifk_2\ \io\wedge \ifk_1^{\pre}\not\sim\ifk_2\big]>0.
        \label{eq:e-not-sim-e}
    \end{align}
  \end{itemize}

It is not hard to see that $(\ast)$ is \msonab-definable, see Appendix~\ref{app:definability}. For \eqref{eq:e-not-sim-e} use Lemma~\ref{lem:c1pre}. 

Roughly, the intuition behind this proposition is as follows. The statement of the proposition can be equivalently written as: there exist two numbers $\ell_1>\ell_2$ such that with nonzero probability $\ifj$ alternates between them. But this property is hard to express in our logic; it requires counting to make sure that $\ell_1>\ell_2$. To remedy this difficulty we make use of Lemma~\ref{lem:msoulemma}. This lemma provides us with an important equivalence between a property that is hard to express (a) $\ell_1>\ell_2$ and a property that we can express in our logic more easily: (b) there exists a vector sequence of dimension $\ell_1$ that is not an asymptotic mix of any vector sequence of dimension~$\ell_2$. 

We start with an~explanation of $(\ast)$ and then proceed to give a sketch of the proof. The~complete proof can be found in Appendix~\ref{app:proof-of-prop}.

The families of intervals $\ifj_1\subseteq\ifj$ and $\ifj_2\subseteq\ifj$ are meant to represent two families of eventually constant intervals of two distinct lengths $\ell_1>\ell_2$, as in Lemma~\ref{lem:d1d2}. Once $\ifj_1$ and $\ifj_2$ are fixed, the families~$\ifi_1$ and $\ifi_2$ are defined uniquely as the families of those intervals in $\ifi$ that are wrapped by some intervals in $\ifj_1$ and $\ifj_2$ respectively. With $\ifi'_1\leq \ifi_1$  we will imitate the vector sequence $\vsf$ of dimension $\ell_1$ that is not an asymptotic mix of any vector sequence $\vsg$ of dimension $\ell_2$ (it exists because of Lemma~\ref{lem:msoulemma}). The rest of~$(\ast)$ expresses that $\vsf$ is not an asymptotic mix of $\vsg$. Thus, $\ifk_1$ represents a~choice of $f\in\vsf$, while $\ifk_2$ represents a~choice of $g\in\vsg$. Finally, the last line of~$(\ast)$ (see~\eqref{eq:e-not-sim-e}) says that $f\not\sim g$. Note here that, the fact that $\ifj_1$ tail\=/precedes $\ifj_2$ implies that $\ifk_1$ precedes $\ifk_2$, so $\ifk_1^\pre$ is well\=/defined.


\noindent{($\Rightarrow$)} The idea for the forward implication follows the explanation given above. We construct $\ifj_1$, $\ifj_2$ of respective
lengths $\ell_1$ and $\ell_2$ using Lemma~\ref{lem:d1d2}. From Lemma~\ref{lem:msoulemma}, we set $\vsf$ to be a vector
sequence of dimension $\ell_1$ that is not an~asymptotic mix of any vector sequence of dimension $\ell_2$.  The
assumption that $\pr\big[\ifi\ \io\ \Rightarrow\ (\liminf \ifi=\infty)\big]=1$ guarantees that the intervals in $\ifi_1$ and
$\ifi_2$ are \emph{long}, so with $\ifi_1'\leq\ifi_1$ we are able imitate the vector sequence $\vsf$ while the choice of
$\ifi_2'$ represents a~vector sequence $\vsg$.

At this point, to facilitate (see Remark~\ref{rem:whyseparation}) the construction of $\ifk_1$ we use the equivalence
between separation and asymptotic mixes described in Lemma~\ref{lem:separates}. The proof is finalized by doing a case
analysis of the two cases in the definition of separation: Definition~\ref{def:separation}. Depending on the case, we
fix the extraction $\ifk_1$ either by picking intervals of length as small (in the first case) or as big (in the latter
case) as possible from $\ifi_1'$.

\noindent{($\Leftarrow$)} The converse implication is easier, it relies on \emph{copying}.  We assume that almost surely whenever $\ifj$ appears infinitely often then it is eventually constant (the negation of the first statement) and use this to refute the second
statement. This is done by copying in the following sense. When $\ifi_1'\leq \ifi_1$ is fixed, we find a~family
$\ifi_2'\leq \ifi_2$ that copies the choice made in $\ifi_1'$; and the same for restrictions $\ifk_2$ based on
$\ifk_1$. In the end, in almost every branch we will have number sequences that are asymptotically equivalent, refuting
the last line in~\eqref{eq:e-not-sim-e}. This terminates the (sketch of the) proof of Proposition~\ref{pro:constant}.

It is not hard to remove the assumption in Proposition~\ref{pro:constant} so as to get Theorem~\ref{thm:const-def}. It suffices to quantify existentially over the wrapped interval $\ifi$ and make sure that $\ifj$ is \emph{sufficiently spaced}. The details can be found in Appendix~\ref{ap:implicit wrappings}. 


\section{Reducing two-counter machines with zero tests}
\label{sec:two-counter}
A two-counter machine has a finite set of control states and two counters, which can be increased, decreased, and tested for zero. The question of whether such a machine has a halting run, is undecidable. In this section we will demonstrate that given a two-counter machine $\mathcal{M}$, we can effectively construct a formula $\Phi(\mathcal{M})$ such that $\mathcal{M}$ has a halting run if and only if $\Phi(\mathcal{M})$ is true.

The reduction is relatively standard once equipped with Theorem~\ref{thm:const-def}. The reason being that Theorem~\ref{thm:const-def} already allows us to do arithmetic in an asymptotic sense: Suppose that $\ifi_1,\ifi_2$ are two families of intervals that are isolated (that is $\bigcup\ifi_1 \cap \bigcup\ifi_2=\emptyset$) and eventually constant on almost every branch. This means that for almost every branch $\pi$, $\ifi_i(\pi)$ defines some natural number $L_i(\pi)$. With the help of Theorem~\ref{thm:const-def} we can express, for instance, that for almost every branch $\pi$, $L_1(\pi)=L_2(\pi)+1$, or that $L_i(\pi)=0$ as follows.

\begin{lemma}
\label{lem:tsts}
Let $\ifi_1$, $\ifi_2$ be isolated families of intervals that appear infinitely often and are eventually constant almost surely. For almost every branch $\pi$, $\ifi_i(\pi)$ is eventually constant, equal to some number, say $L_i(\pi)$. Then, one can express in \msonab the following:
\begin{align}
\pr\big[L_1=0\big]&=1\label{eq:tst-zero}\\
\pr\big[L_1=L_2+1\big]&=1.\label{eq:tst-inc}
\end{align}
\end{lemma}
\begin{proof}
  Condition~\eqref{eq:tst-zero} is directly formalisable in \msonab. The formula says that for almost every branch $\pi$, after some threshold, every node in $\source{\ifi_1}\cap\pi$ has a child that is in $\target{\ifi_1}$. 

  As for Condition~\eqref{eq:tst-inc}, first we can easily express that almost surely $L_1(\pi)>0$ (this is a~necessary condition for~\eqref{eq:tst-inc}). If this is the case, then we define another family $\ifi_3$, such that $\source{\ifi_3}=\source{\ifi_1}$  and the targets $\target{\ifi_3}$ are exactly the parents of the nodes in $\target{\ifi_1}$, \ie we move the targets of $\ifi_1$ to their parents thereby decreasing the lengths of intervals by~1. The set $\target{\ifi_3}$ (and therefore~$\ifi_3$) is \mso-definable and moreover for $L_3(\pi)$ defined analogously, we have $L_3(\pi)=L_1(\pi)-1$ for almost every branch~$\pi$.

  Thus, to verify that $L_3(\pi)=L_2(\pi)$, \ie Condition~\eqref{eq:tst-inc}, it is enough to check that $\ifi_3\cup\ifi_2$ is eventually constant, almost surely by applying Theorem~\ref{thm:const-def}. This is possible because $\ifi_3$ and $\ifi_2$ are disjoint and therefore $\ifi_3\cup\ifi_2$ is a valid family of intervals.
\end{proof}

We now illustrate how a run of $\mathcal{M}$ is encoded. It is of the form
\[
  (q_1,c^1_1,c^2_1), (q_2,c^1_2,c^2_2), \ldots , (q_\ell,c^1_\ell,c^2_\ell),
\]
where $q_k$ are control states and $c^1_k,c^2_k$ is the value of the first and second counter on the $k$th step respectively.
We will encode such a run using three families of intervals: $\ifj,\ifi_1,\ifi_2$ and a labelling by states. $\ifj$ will be eventually constant and equal to $\ell$ (the length of the run), and the nodes in $\inter(\ifj)$ will be labeled by the control states of $\mathcal{M}$. In other words, intervals $[x,y]\in\ifj$ will be such that $\inter([x,y])=\set{x_1 < \cdots < x_\ell}$ and $x_k$ is labeled by $q_k$. Further, $\ifj$ will wrap both $\ifi_1$ and $\ifi_2$, and $x_k$ will be the source of an interval in $\ifi_1$ of length $c_k^1$ (\ie $\ifi_1(x_k)=c_k^1$), as well as the source of an interval in $\ifi_2$ of length $c_k^2$ (\ie $\ifi_2(x_k)=c_k^2$). 
\begin{example}
  Consider the run:
  \[
    (q_0,0,0), (q_1,0,1), (q_2,1,2). 
  \]
  Its encoding with intervals is depicted in Figure~\ref{fig:exampleencoding}. 
  \begin{figure}
    \includegraphics[width=0.5\textwidth]{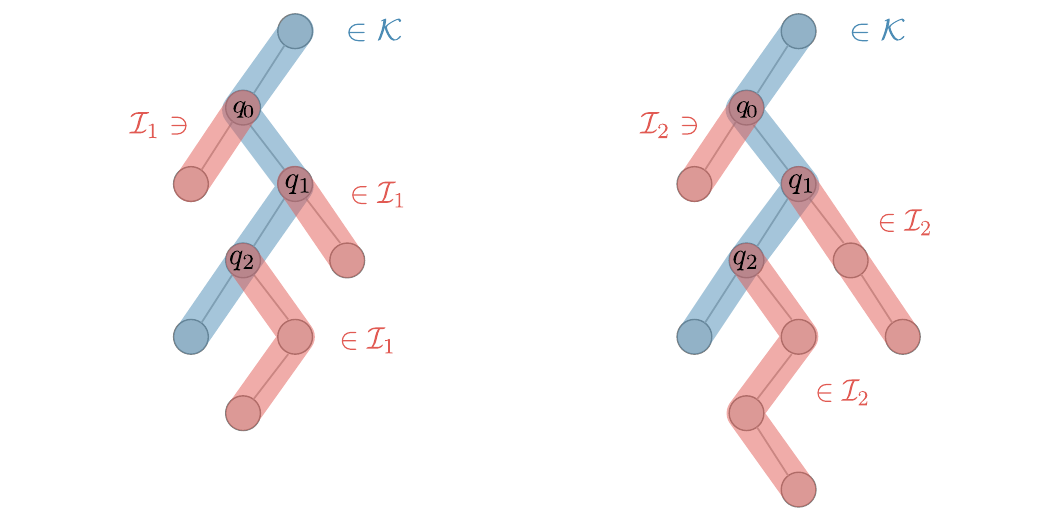}
    \caption{$\inter(\ifj)$ is labeled by the states $q_i$. Note the lengths of the intervals in $\ifi_i$. On the left
      we have the respective lengths 0,0 and 1 corresponding to the history of counter 1 in the run. On the right: 0,1
      and 2.}
    \label{fig:exampleencoding}
  \end{figure}
\end{example}
In order to ascertain that a run is valid, we need to check whether the counters are being increased and decreased correctly using Lemma~\ref{lem:tsts}. For this purpose it is necessary to be able to speak about, for instance, the value of counter 1 at step $k$ by choosing the correct subset of $\ifi_1$. This can be done as follows.

If $\ifi'_1\subseteq\ifi_1$ is a~family of intervals that is an~extraction of $(\ifj,\ifi_1)$, we say that $\ifj'\leq\ifj$ is \emph{induced} by $\ifi'_1$ if $\target{\ifj'}=\source{\ifi'_1}$. We say that $\ifi'_1$ is a~\emph{component selector of $\ifi_1$} if $\ifj'$ induced by $\ifi'_1$ is eventually constant with probability~$1$. In such a~case, the lengths of the intervals in $\ifi'_1$ (from some moment on, along almost every branch of the tree) correspond to the values of the counter 1 at a certain step. In other words, $\ifi'_1$ is a~component selector if on almost every branch $\pi$, there exists a~number $k\in\nat$ such that $\ifi'_1$ is eventually choosing exactly the $k$th component.

The following proposition follows directly from the ability to express Conditions~\eqref{eq:tst-zero} and~\eqref{eq:tst-inc}.

\begin{proposition}
  \label{prop:two-counter}
  For every two\=/counter machine with zero tests $\mathcal{M}$, we can effectively compute a~formula $\phi(\mathcal{M})$ of
  \msonab{}, such that $\phi(\mathcal{M})$ is true if and only if $\mathcal{M}$ halts.
\end{proposition}
\begin{proof}
  The first part of the formula $\phi(\mathcal{M})$ says: there exist families of intervals $\ifj$, $\ifi_1$, and $\ifi_2$ and a~labelling $\rho$ of $\inter(\ifj)$ by states of $\mathcal{M}$ such that:
  \begin{itemize}
  \item $\ifj$ appears infinitely often and is eventually constant almost surely,
  \item $\ifj$ wraps both $\ifi_1$ and $\ifi_2$, and
  \item every component selector $\ifi'$ of either $\ifi_1$ or $\ifi_2$ on almost every branch is eventually constant and the labels of $\rho$ in the nodes $\source{\ifi'}$ stabilise almost surely.
  \end{itemize}

  This implies that for $i=1,2$ and almost every branch $\pi$, $\vec{\ifj}(\ifi_i, \pi)$ is a~vector sequence that is eventually constant
  equal to some vector $(c^{i}_1,c^{i}_2,\ldots,c^{i}_\ell)(\pi)$. Moreover, on almost every branch $\pi$ the labels of the nodes in $\inter(\ifj)$ must also stabilise to some sequence~$(q_1,\ldots,q_\ell)(\pi)$.
    
  The second part of the formula uses component selectors as well as the conditions from Lemma~\ref{lem:tsts} to test the relationship between the values $(c^1_n, c^2_n, q_n, c^1_{n+1}, c^2_{n+1}, q_{n+1})(\pi)$ to verify that on almost every branch $(c^1_1,\ldots,c^1_\ell)(\pi)$, $(c^2_1,\ldots,c^2_\ell)(\pi)$, and $(q_1,\ldots,q_\ell)(\pi)$ is a~valid run of~$\mathcal{M}$. This is done by requiring that the values of counters and the labeling in any two consecutive component selectors respect the transition relation of~$\mathcal{M}$. 
  
  If the formula is true then the witnessing families $\ifj$, $\ifi_1$, $\ifi_2$, and a~labelling $\rho$ must almost surely encode (the unique) accepting run of $\mathcal{M}$. Conversely, if $\mathcal{M}$ has an~accepting run then one can easily choose families as above such that each interval $[x,y]\in\ifj$ encodes in fact this single run. This implies that the above \msonab formula must be true in that case.
\end{proof}
Theorem~\ref{thm:main} is a corollary of Proposition~\ref{prop:two-counter}. 


\section{Conclusions}
\label{sec:conclusions}

The undecidability result from this paper, together with the undecidability results about \msou  from~\cite{bojanczyk_msou_final,u1u2}, lead to the following fundamental  question: is there \emph{any} quantifier that can be added to {\sc mso} on infinite words (or trees), while retaining decidability?  Of course a negative answer would require formalising what ``quantifier'' means. A natural direction is to use  the abstract approach from~\cite{lohrey2014boolean}, which precludes positive answers that involve adding unary predicates as discussed in~\cite{rabinovich2006decidable}. 

\section*{Acknowledgments}
\addcontentsline{toc}{section}{Acknowledgments}

The first two authors have been supported by ERC Consolidator grant LIPA 683080. The last author has been supported by Polish National Science Centre grant 2016/22/E/ST6/00041.

\bibliographystyle{plain}
\bibliography{bibliography}

\begin{thebibliography}{10}

\bibitem{baier2012probabilistic}
Christel Baier, Marcus Gr{\"o}{\ss}er, and Nathalie Bertrand.
\newblock Probabilistic $\omega$-automata.
\newblock {\em Journal of the ACM (JACM)}, 59(1):1, 2012.

\bibitem{baier1998model}
Christel Baier and Marta Kwiatkowska.
\newblock Model checking for a probabilistic branching time logic with
  fairness.
\newblock {\em Distributed Computing}, 11(3):125--155, 1998.

\bibitem{berthon19_monad_secon_order_logic_with}
Rapha{\"e}l Berthon, Emmanuel Filiot, Shibashis Guha, Bastien Maubert, Aniello
  Murano, Jean-Fran{\c{c}}ois Raskin, and Sasha Rubin.
\newblock Monadic second-order logic with path-measure quantifier is
  undecidable.
\newblock {\em CoRR}, 2019.
\newblock \url{https://arxiv.org/abs/1901.04349}.

\bibitem{berthon2017threshold}
Rapha{\"e}l Berthon, Mickael Randour, and Jean-Fran{\c{c}}ois Raskin.
\newblock {Threshold Constraints with Guarantees for Parity Objectives in
  Markov Decision Processes}.
\newblock In Ioannis Chatzigiannakis, Piotr Indyk, Fabian Kuhn, and Anca
  Muscholl, editors, {\em 44th International Colloquium on Automata, Languages,
  and Programming (ICALP 2017)}, volume~80 of {\em Leibniz International
  Proceedings in Informatics (LIPIcs)}, pages 121:1--121:15, Dagstuhl, Germany,
  2017. Schloss Dagstuhl--Leibniz-Zentrum fuer Informatik.

\bibitem{bojanczyk2016thin}
Miko{\l}aj Boja{\'{n}}czyk.
\newblock {Thin MSO with a Probabilistic Path Quantifier}.
\newblock In Ioannis Chatzigiannakis, Michael Mitzenmacher, Yuval Rabani, and
  Davide Sangiorgi, editors, {\em 43rd International Colloquium on Automata,
  Languages, and Programming (ICALP 2016)}, volume~55 of {\em Leibniz
  International Proceedings in Informatics (LIPIcs)}, pages 96:1--96:13,
  Dagstuhl, Germany, 2016. Schloss Dagstuhl--Leibniz-Zentrum fuer Informatik.

\bibitem{u1u2}
Miko{\l}aj Boja{\'{n}}czyk, Laure Daviaud, Bruno Guillon, Vincent Penelle, and
  A.~V. Sreejith.
\newblock Undecidability of {MSO}+ultimately periodic (preprint).

\bibitem{bojanczyk2017emptiness}
Miko{\l}aj Boja{\'{n}}czyk, Hugo Gimbert, and Edon Kelmendi.
\newblock {Emptiness of Zero Automata Is Decidable}.
\newblock In Ioannis Chatzigiannakis, Piotr Indyk, Fabian Kuhn, and Anca
  Muscholl, editors, {\em 44th International Colloquium on Automata, Languages,
  and Programming (ICALP 2017)}, volume~80 of {\em Leibniz International
  Proceedings in Informatics (LIPIcs)}, pages 106:1--106:13, Dagstuhl, Germany,
  2017. Schloss Dagstuhl--Leibniz-Zentrum fuer Informatik.

\bibitem{fullpaper}
Miko{l}aj Boja{'{n}}czyk, Edon Kelmendi, and {Micha{l} Skrzypczak}.
\newblock Mso+nabla is undecidable.
\newblock {\em CoRR}, abs/1901.06900, 2019.
\newblock \url{https://arxiv.org/abs/1901.06900}.

\bibitem{bojanczyk_msou_final}
Miko{\l}aj Boja{\'{n}}czyk, Pawe{\l} Parys, and Szymon Toru{\'{n}}czyk.
\newblock {The MSO+U Theory of (N,<) Is Undecidable}.
\newblock In Nicolas Ollinger and Heribert Vollmer, editors, {\em 33rd
  Symposium on Theoretical Aspects of Computer Science (STACS 2016)}, volume~47
  of {\em Leibniz International Proceedings in Informatics (LIPIcs)}, pages
  21:1--21:8, Dagstuhl, Germany, 2016. Schloss Dagstuhl--Leibniz-Zentrum fuer
  Informatik.

\bibitem{brazdil2008satisfiability}
Tom{\'a}{\v{s}} Br{\'a}zdil, Vojtech Forejt, Jan Kret{\'\i}nsk{\`y}, and
  Anton{\'\i}n Kucera.
\newblock The satisfiability problem for probabilistic {CTL}.
\newblock In {\em 2008 23rd Annual IEEE Symposium on Logic in Computer
  Science}, pages 391--402. IEEE, 2008.

\bibitem{carayol2014randomization}
Arnaud Carayol, Axel Haddad, and Olivier Serre.
\newblock Randomization in automata on infinite trees.
\newblock {\em ACM Transactions on Computational Logic (TOCL)}, 15(3):24, 2014.

\bibitem{gimbert2010probabilistic}
Hugo Gimbert and Youssouf Oualhadj.
\newblock Probabilistic automata on finite words: Decidable and undecidable
  problems.
\newblock In {\em International Colloquium on Automata, Languages, and
  Programming}, pages 527--538. Springer, 2010.

\bibitem{hansson1994logic}
Hans Hansson and Bengt Jonsson.
\newblock A logic for reasoning about time and reliability.
\newblock {\em Formal Aspects of Computing}, 6(5):512--535, Sep 1994.

\bibitem{hart1986probabilistic}
Sergiu Hart and Micha Sharir.
\newblock Probabilistic propositional temporal logics.
\newblock {\em Information and Control}, 70(2-3):97--155, 1986.

\bibitem{lehman_time_and_chance}
Daniel Lehmann and Saharon Shelah.
\newblock Reasoning with time and chance.
\newblock {\em Information and Control}, 53(3):165--198, 1982.

\bibitem{lohrey2014boolean}
Markus Lohrey and Georg Zetzsche.
\newblock On boolean closed full trios and rational kripke frames.
\newblock In {\em 31st International Symposium on Theoretical Aspects of
  Computer Science (STACS 2014)}. Schloss Dagstuhl-Leibniz-Zentrum fuer
  Informatik, 2014.

\bibitem{michalewski2016measure}
Henryk Michalewski and Matteo Mio.
\newblock Measure quantifier in monadic second order logic.
\newblock In {\em International Symposium on Logical Foundations of Computer
  Science}, pages 267--282. Springer, 2016.

\bibitem{paz1971introduction}
Azaria Paz.
\newblock {\em Introduction to Probabilistic Automata (Computer Science and
  Applied Mathematics)}.
\newblock Academic Press, Inc., Orlando, FL, USA, 1971.

\bibitem{rabin_s2s}
Michael~Oser Rabin.
\newblock Decidability of second-order theories and automata on infinite trees.
\newblock {\em Trans. of the American Math. Soc.}, 141:1--35, 1969.

\bibitem{rabinovich2006decidable}
Alexander Rabinovich and Wolfgang Thomas.
\newblock Decidable theories of the ordering of natural numbers with unary
  predicates.
\newblock In {\em International Workshop on Computer Science Logic}, pages
  562--574. Springer, 2006.

\bibitem{thomas_languages}
Wolfgang Thomas.
\newblock Languages, automata, and logic.
\newblock In {\em Handbook of Formal Languages}, pages 389--455. Springer,
  1996.

\bibitem{vardi1985automatic}
Moshe~Y Vardi.
\newblock Automatic verification of probabilistic concurrent finite state
  programs.
\newblock In {\em 26th Annual Symposium on Foundations of Computer Science
  (SFCS 1985)}, pages 327--338. IEEE, 1985.

\bibitem{vardi1986automata}
Moshe~Y Vardi and Pierre Wolper.
\newblock An automata-theoretic approach to automatic program verification.
\newblock In {\em Proceedings of the First Symposium on Logic in Computer
  Science}, pages 322--331. IEEE Computer Society, 1986.

\end{thebibliography}
\nocite{fullpaper}
\clearpage 
\appendix

\subsection{Proofs in Section~\ref{sec:econst}}
\subsubsection{Proof of Lemma~\ref{lem:c1pre}}
\label{app:lemc1pre}
\begin{replemma}{lem:c1pre}
  \lemcpre{}
\end{replemma}
    \begin{proof}
      We claim that the property in the statement of the lemma is equivalent to
      \begin{itemize}
      \item[$(\ast)$] either
        \begin{itemize}
        \item[$(\dagger_1)$] there exists $\ifi\subseteq\pre(\ifi_2)$ such that
          \begin{align*}
            \pr\Bigg[\wedge\begin{cases}
              \suc(\ifi)\ \io\\
              \limsup\ifi < \infty\\
              \limsup\suc(\ifi) = \infty
            \end{cases}\Bigg]>0, \text{ or}
          \end{align*}
          \item[$(\dagger_2)$] there exists $\ifi\subseteq\ifi_2$ such that
          \begin{align*}
            \pr\Bigg[\wedge\begin{cases}
              \ifi\ \io\\
              \limsup\ifi < \infty\\
              \limsup\pre(\ifi) = \infty
            \end{cases}\Bigg]>0.
          \end{align*}
        \end{itemize}
      \end{itemize}
      The property $(\ast)$ is \msonab-definable since $\pre$ and $\suc$ are \mso-definable and for checking the boundedness we can use Lemma~\ref{lem:char}. See Appendix~\ref{app:definability}.
      
\noindent ($\Rightarrow$) Assume that there exists a~set of branches $R\subseteq \two^\omega$ that has a~non\=/zero probability,
  such that for each branch $\pi\in R$, $\ifi_2$ appears infinitely often in $\pi$ and there exists a~set of positions $X_\pi\subseteq \nat$ on which the sequence $\ifi_1^{\pre}(\pi)$ is bounded but
  $\ifi_2(\pi)$ is not (the dual case is analogues, see below). By $\aleph_0$\=/additivity of the measure, this implies that there exists $b\in\nat$ such that:
  \begin{align}
    \pr\Bigg[\begin{cases}
      \ifi_2\ \io, \text{ and}\\
      \exists X\subseteq \nat. \begin{cases}
        \ifi_1^\pre\restriction_X\equiv b, \text{ and}\\
        \limsup\ifi_2\restriction_X=\infty
      \end{cases}
    \end{cases}\Bigg]>0.
    \label{eq:existc-one}
  \end{align}


  Take $\ifi\subseteq\pre(\ifi_2)$ as the family of intervals that have length equal to $b$. Take any branch $\pi$ in
  the set from~\eqref{eq:existc-one} and let $X_\pi\subseteq\nat$ be a~witness. Clearly, $X_\pi$ must be infinite and therefore $\ifi$ appears infinitely often in $\pi$  and $\limsup\ifi(\pi)= b<\infty$. On the other hand, $\limsup\suc(\ifi)(\pi)=\infty$ because $\suc(\ifi)(\pi)$ contains as a~subsequence the lengths of intervals in~$\ifi_2$ that are measured in~$\ifi_2\restriction_{X_\pi}$, see Remark~\ref{rem:suc-correspondence}. It means in particular that $\suc(\ifi)$ appears infinitely often in~$\pi$. Therefore, $(\dagger_1)$ holds for $\ifi$ and such~$\pi$, which means that the probability there is positive.
  
  In the dual case, when for each $\pi\in R$ there is~$X_\pi$ such that sequence $\ifi_1^{\pre}(\pi)\restriction_{X_\pi}$ is unbounded but $\ifi_2(\pi)\restriction_{X_\pi}$ is bounded, we know that there exists $b\in\nat$ such that:
  \begin{align}
    \pr\Bigg[\begin{cases}
      \ifi_2\ \io, \text{ and}\\
      \exists X\subseteq \nat. \begin{cases}
        \limsup\ifi_1^\pre\restriction_X=\infty, \text{ and}\\
        \ifi_2\restriction_X\equiv b
      \end{cases}
    \end{cases}\Bigg]>0.
    \label{eq:existc-two}
  \end{align}
  In that case we take $\ifi\subseteq\ifi_2$ as the family of intervals of length equal to $b$. For each branch $\pi$ in the set from~\eqref{eq:existc-two} and its witness $X_\pi$ we have: $\ifi\ \io$ in $\pi$; $\limsup\ifi(\pi)=b <\infty$; and $\limsup\pre(\ifi)(\pi)=\infty$ --- notice that the sequence $\pre(\ifi)(\pi)$ contains the sequence $\ifi_1^\pre(\pi)\restriction_{X_\pi}$ as, possibly strict, subsequence. However, as the latter is unbounded, also the former must be unbounded. Therefore, $(\dagger_2)$ holds.
  
  \noindent ($\Leftarrow$) Assume that~$(\dagger_1)$ is true and fix $\ifi\subseteq\pre(\ifi_2)$.  Take any
  branch $\pi$ in the set measured in~$(\dagger_1)$. Since $\suc(\ifi)$ appears infinitely often in $\pi$, by the definition of
  $\suc$ we have $\ifi_2$ also appears infinitely often in $\pi$. We will show that $\ifi_1^\pre(\pi)\not\sim\ifi_2(\pi)$.
  
  For $k\in\nat$, denote by $\mathrm{source}_k(\ifi_2)$ the set of sources of $\ifi_2$ that have exactly $k-1$ strict ancestors that are also sources of $\ifi_2$.  Let $X_\pi$ be the set of numbers~$k$ such that $\pi\cap\mathrm{source}_k(\ifi_2)\cap \source{\suc(\ifi)}\neq\emptyset$. Then $\ifi_2(\pi)\restriction_{X_\pi}=\suc(\ifi)(\pi)$ is unbounded by the assumption. On the other hand, \added{$\ifi\subseteq\pre(\ifi_2)$ and by the definition of $X_\pi$ we know that} $\ifi_1^\pre(\pi)\restriction_{X_\pi}$ is a~subsequence of $\ifi(\pi)$ and is therefore bounded. This concludes the proof that $\ifi_1^\pre(\pi)\not\sim\ifi_2(\pi)$.
  
  Finally, consider the last case that~$(\dagger_2)$ holds and fix $\ifi\subseteq\ifi_2$ witnessing that. Take a~branch $\pi$ from the set measured in~$(\dagger_2)$. The fact that $\ifi$ appears infinitely often in $\pi$ implies directly that $\ifi_2$ also appears infinitely often in in $\pi$. Take $X_\pi$ as the set of numbers $k$ such that $\pi\cap\mathrm{source}_k(\ifi_2)\cap\source{\ifi}\neq\emptyset$. Then $\ifi_2(\pi)\restriction_{X_\pi}=\ifi(\pi)$ is bounded. However, $\ifi_1^\pre(\pi)\restriction_{X_\pi}$ contains $\pre(\ifi)(\pi)$ as a~subsequence and therefore is unbounded. Thus, $\ifi_1^\pre(\pi)\not\sim\ifi_2(\pi)$.
\end{proof}

\subsubsection{Proof of Proposition~\ref{pro:constant}}
\label{app:proof-of-prop}

This section of the appendix is devoted to the~proof of Proposition~\ref{pro:constant}.

\begin{repproposition}{pro:constant}
\proconstant{}
\end{repproposition}

The claim is that the property in the statement of the proposition is equivalent to:
\begin{itemize}
  \item[$(\ast)$] there exist isolated  $\ifj_1,\ifj_2\subseteq\ifj$, where $\ifj_1$ tail\=/precedes $\ifj_2$, $\pr[\ifj_2\ \io]>0$, and if
    $\ifi_1,\ifi_2\subseteq \ifi$ are such that $\ifj_i$ wraps $\ifi_i$, $i\in\set{1,2}$ then:
    \begin{align}
      &\exists\ifi_1'\leq\ifi_1.\ \ \forall \ifi_2'\leq\ifi_2.\nonumber\\
      &\ \ \exists\ifk_1\subseteq\ifi_1'\text{ extraction of }(\ifj_1,\ifi_1').\nonumber\\
      &\ \ \forall\ifk_2\subseteq\ifi_2'\text{ extraction of }(\ifj_2,\ifi_2').\nonumber\\
      &\qquad\pr\big[\ifk_2\ \io\wedge \ifk_1^{\pre}\not\sim\ifk_2\big]>0.
        \tag{\ref{eq:e-not-sim-e}}
    \end{align}
  \end{itemize}

\noindent
{\em Proof of the forward implication:} 

Let $\ifj_1$, $\ifj_2$ be as in Lemma~\ref{lem:d1d2}, so that every interval in $\ifj_1$ (respectively $\ifj_2$) has length~$\ell_1$
  (respectively~$\ell_2$), $\ell_1>\ell_2$, $\ifj_1$ tail-precedes $\ifj_2$, and $\pr[\ifj_2\ \io]>0$. Let
  $\ifi_1,\ifi_2\subseteq \ifi$ be such that $\ifj_i$ wraps $\ifi_i$ for $i=1,2$---notice that such $\ifi_1$, $\ifi_2$ are defined uniquely by these conditions.

  Let $\vsf$ be a~vector sequence of dimension $\ell_1$ that is not an~asymptotic mix of any vector sequence of dimension
  $\ell_2$. It exists thanks to Lemma~\ref{lem:msoulemma}.

  We construct $\ifi_1'\leq\ifi_1$ as follows. If $k\in\nat$ and $x_k\in\source{\ifj_1}$ has exactly $k$ strict ancestors in $\source{\ifj_2}$ then:
 \begin{align}
   \vec{\ifj}_1(\ifi_1,x_k)&=(v_1,v_2,\ldots,v_{\ell_1}),\nonumber \\
   \vsf(k) &= (w_1,w_2,\ldots,w_{\ell_1}),\nonumber\\
   \vec{\ifj}_1(\ifi_1',x_k)&=(v_1',v_2',\ldots,v_{\ell_1}'),\label{eq:def-C-1-p}\\
                           &\text{where $v_i'=\min(v_i,w_i)$ for $i=1,2,\ldots,\ell_1$.}\nonumber
 \end{align}
 
 \begin{lemma}
 Assume that $\pi$ is a~branch such that~$\ifj_2$ appears infinitely often in $\pi$ and $\liminf\ifi_1(\pi)=\infty$. Then for every $f\in\vsf$ there exists $f'\in\vec{\ifj}_1^\pre(\ifi_1',\pi)$ such that $f'\sim f$.
 In particular, $\vec{\ifj}_1^\pre(\ifi_1',\pi)$ is not an~asymptotic mix of any vector sequence of dimension strictly smaller than~$\ell_1$.
 \end{lemma}
 
 \begin{proof}
  Fix some $f\in\vsf$. Notice that for $k\in\nat$ the vector $\vec{\ifj}_1^\pre(\ifi_1',\pi)(k)$ is given by the formula~\eqref{eq:def-C-1-p}. Thus, we can construct $f'\in\vec{\ifj}_1^\pre(\ifi_1',\pi)$ by copying~$f$. More formally, for all $k\in\nat$, if $f(k)$ is the $i$th component of $\vsf$ then also $f'(k)$ is the $i$th component of~ $\vec{\ifj}_1^\pre(\ifi_1',\pi)(k)$.
  
  We prove that $f'\sim f$. Let $X\subseteq\nat$, and suppose that $f\restriction_X$ is
  bounded. Then $f'\restriction_X$ is bounded as well, since construction, we have that for all $n\in\nat$,
  $f'(n)\leq f(n)$ (see~\eqref{eq:def-C-1-p}). If on the other hand $f\restriction_X$
  is unbounded, then so is $f'\restriction_X$, as a~consequence of the fact that $\lim\ifi_1(\pi)=\infty$.
  
  Now assume that $\vec{\ifj}_1^\pre(\ifi_1',\pi)$ is an~asymptotic mix of a~vector sequence $\vsg$ of dimension strictly smaller than $\ell_1$. In that case $\vsf$ must be an~asymptotic mix of~$\vsg$: for each $f\in\vsf$ there exists $f'\in\vec{\ifj}_1^\pre(\ifi_1',\pi)$ given by the above construction such that $f\sim f'$; moreover by assumption there exists $g\in\vsg$ such that $f'\sim g$; and thus $f\sim g$; a~contradiction.
 \end{proof}

 Fix some $\ifi_2'\leq\ifi_2$ and take a~branch $\pi$ on which $\ifj_2$ appears infinitely often and $\liminf\ifi_1(\pi)=\infty$
 (the assumptions on $\ifj_2$ and $\ifi$ guarantee that with a~positive probability a~random branch has these properties).
 By the lemma above $\vec{\ifj}_1^\pre(\ifi_1',\pi)$ is not an~asymptotic mix of $\vec{\ifj}_2(\ifi_2',\pi)$. This means that we have:
 \begin{align*}
   \pr\Bigg[\begin{cases}\ifj_2\ \io, \text{ and}\\
     \vec{\ifj}_1^\pre(\ifi_1')\text{ is not an~asymp. mix of }\vec{\ifj}_2(\ifi_2')
     \end{cases}\hspace{-12pt}\Bigg]>0.
 \end{align*}
 Lemma~\ref{lem:separates} implies that
 \begin{align*}
   \pr\Bigg[\begin{cases}
     \ifj_2\ \io, \text{ and}\\
     \exists b\in\nat.\ \text{$b$ separates $\vec{\ifj}_1^\pre(\ifi_1')$ from $\vec{\ifj}_2(\ifi_2')$}
   \end{cases}\hspace{-12pt}\Bigg]>0.
 \end{align*}
 And thus, by countable additivity of measures, there must exist $b\in\nat$ such that:
 \begin{align*}
   \pr\big[\ifj_2\ \io\ \wedge \text{$b$ separates $\vec{\ifj}_1^\pre(\ifi_1')$ from $\vec{\ifj}_2(\ifi_2')$}\big]>0. 
 \end{align*}
 From the definition of separation we now have the following two cases:
 \begin{align}
   \label{eq:case1}
   \pr\Bigg[\begin{cases}
     \ifj_2\ \io, \text{ and}\\
     \exists X.\begin{cases}
       \min\big(\vec{\ifj}_1^\pre(\ifi_1')\restriction_X\big)\leq b,\text{ and}\\
       \text{$\min\big(\vec{\ifj}_2(\ifi_2')\restriction_X\big)$ is unbnd.}
     \end{cases}
   \end{cases}
   \hspace{-25pt}\Bigg]>0,
 \end{align}
 and
 \begin{align}
   \label{eq:case2}
   \pr\Bigg[\begin{cases}
     \ifj_2\ \io, \text{ and}\\
     \exists X.\begin{cases}
       \max\big(\vec{\ifj}_2(\ifi_2')\restriction_X\big)\leq b, \text{ and}\\
       \text{$\max\big(\vec{\ifj}_1^\pre(\ifi_1')\restriction_X\big)$ is unbnd.}
     \end{cases}
   \end{cases}
   \hspace{-23pt}\Bigg]>0.
 \end{align}
 
 \noindent
 {\em The first case:}
 
 Construct an~extraction $\ifk_1\subseteq\ifi_1'$ of $(\ifj_1,\ifi_1')$ by picking any interval whose length is smaller
 than $b$ (if there is none, we pick arbitrarily). We fix an~extraction $\ifk_2\subseteq\ifi_2'$ of $(\ifj_2,\ifi_2')$,
 and prove that
 \begin{align*}
   \pr\big[\ifk_2\ \io\ \wedge\ifk_1^\pre\not\sim\ifk_2\big]>0.
 \end{align*}
 Since $\ifj_1$ precedes $\ifj_2$ (tail-preceding is a~stronger property), we know that for $x\in\source{\ifj_2}$, $\pre(x)$ is well-defined, it is the first ancestor of $x$ in $\source{\ifj_1}$. Let $\ifj_2'\subseteq\ifj_2$ be the family on which we keep only those intervals $[x,y]\in\ifj_2$ such that $\vec{\ifj}_1(\ifi_1',\pre(x))$ has a~component that is smaller than $b$. Then~\eqref{eq:case1} implies that $\pr[\ifj_2'\ \io]>0$.

 For $x\in\source{\ifj_2'}$ define $M(x)$ to be the minimal component in the vector $\vec{\ifj}_2'(\ifi_2',x)$. On a~branch~$\pi$ where $\ifj_2'$ appears infinitely often, there are infinitely many nodes $x_0<x_1<\ldots$ belonging to~$\source{\ifj_2'}$;
 define:
 \[
   M(\ifj_2')(\pi)=M(x_0),M(x_1),\ldots\in\nat^\omega.
 \]
 Then~\eqref{eq:case1} implies that:
 \[
   \pr\big[\ifj_2'\ \io\ \wedge\  \limsup M(\ifj_2')=\infty\big]>0.
 \]
 Finally define $\ifj_2''\subseteq\ifj_2'$ to be the record breakers with respect to the function $M$, \ie~for all
 $x,x'\in\source{\ifj_2''}$, if $x<x'$ then $M(x)<M(x')$. From the inequality above it follows that:
 \begin{equation}
   \label{eq:iid2toinfty}
   \pr\big[\ifj_2''\ \io\ \wedge\  (\liminf M(\ifj_2'') =\infty)\big]>0,
 \end{equation}
 where $M(\ifj_2'')(\pi)$ is the number sequence resulting from applying $M$ only to the sources of the intervals in~$\ifj_2''$.
 Let $\ifk_2'\subseteq\ifk_2$ be such that every element of $\source{\ifk_2'}$ belongs to some interval in
 $\ifj_2''$. Since the intervals in $\ifj_2''$ have length $\ell_2$, the sources of $\ifk_2'$ are always at a~distance
 smaller than $\ell_2$ than the respective source of $\ifj_2''$: if $x'\in\source{\ifk_2'}$ and $x'\in\inter([x,y])\in\ifj_2''$ then $|x'|-|x|\leq\ell_2$. Therefore, as a~consequence of Lemma~\ref{lem:ob1} and~\eqref{eq:iid2toinfty} we have
 \begin{align*}
   \pr\big[\ifk_2'\ \io\ \wedge\ (\liminf\ifk_2'=\infty)\big]>0.
 \end{align*}
 But by construction, the intervals in $\ifk_2'$ are preceded by intervals in $\ifk_1$ whose length is smaller than $b$,
 hence we have proved that
\begin{align*}
   \pr\big[\ifk_2\ \io\ \wedge\ \ifk_1^\pre\not\sim\ifk_2\big]>0.
\end{align*}

\noindent
{\em The second case:}

 Construct $\ifk_1\subseteq\ifi_1'$ extraction of $(\ifj_1,\ifi_1')$ by picking intervals with
 the maximal length. We fix an~extraction $\ifk_2\subseteq\ifi_2'$ of $(\ifj_2,\ifi_2')$, and prove that
 \begin{align*}
   \pr[\ifk_2\ \io\ \wedge\ \ifk_1^\pre\not\sim\ifk_2]>0.
 \end{align*}
 Let $\ifj_2'\subseteq\ifj_2$ be
 the family that keeps only those $[x,y]\in\ifj_2$ for which $\vec{\ifj}_2(\ifi_2',x)$ has all components smaller than $b$. Then~\eqref{eq:case2}
 implies that $\pr[\ifj_2'\ \io]>0$. Let $\ifk_2'\subseteq\ifk_2$ be such that every source of an~interval in
 $\ifk_2'$ belongs to an~interval in $\ifj_2'$. Since the intervals in $\ifj_2'$ all have length $\ell_2$, the distance
 between a~node in $\source{\ifj_2'}$ and it's first descendant in $\source{\ifk_2'}$ is at most $\ell_2$, so applying
 Claim~\ref{lem:ob1} we have that $\pr[\ifk_2'\ \io]>0$. While every interval in $\ifk_2'$ has length at most $b$,
 \eqref{eq:case2} implies that there is a~non\=/zero probability that $\ifk_1^\pre$ is unbounded, \ie
 \begin{align*}
   \pr\big[\ifk_2'\ \io\ \wedge\ \ifk_1^\pre\not\sim\ifk_2'\big]>0. 
 \end{align*}
 
This concludes the proof of the forward implication.

\noindent
{\em Proof of the converse implication:}

Assume that
\begin{align*}
  \pr\big[\ifj\ \io\ \Rightarrow\ \text{$\ifj$ is eventually constant}\big]=1.
\end{align*}
Let $\ifj_1,\ifj_2\subseteq\ifj$ be such that $\ifj_1$ tail\=/precedes $\ifj_2$, and $\pr[\ifj_2\ \io]>0$. Consider $\ifi_1$ and $\ifi_2$ as in the statement and fix $\ifi_1'\leq\ifi_1$.

  We let $\ifi_2'\leq\ifi_2$ be such that for all $x'\in\source{\ifj_2}$ the following holds: let $x=\pre(x')$ (it exists
  because $\ifj_1$ tail-precedes $\ifj_2$) then for all $k\in\nat$ if both $\vec{\ifj}_1(\ifi_1',x)$ and
  $\vec{\ifj}_2(\ifi_2,x')$ have $k$th components defined: $(\vec{\ifj}_1(\ifi_1',x))_k$ and $(\vec{\ifj}_2(\ifi_2,x'))_k$
  then:
  \[
    \big(\vec{\ifj}_2(\ifi_2',x')\big)_k=\min\Bigg\{\big(\vec{\ifj}_2(\ifi_2,x')\big)_k, \big(\vec{\ifj}_1(\ifi_1',x)\big)_k\Bigg\}.
  \]
  When the respective components are not defined, take $\big(\vec{\ifj}_2(\ifi_2',x')\big)_k=\big(\vec{\ifj}_2(\ifi_2,x')\big)_k.$

  Fix $\ifk_1\subseteq\ifi_1'$ an~extraction of $(\ifj_1,\ifi_1')$. We say that $\ifk_1$ \emph{chooses $k$th component in $x$} if $[x,y]\in\ifj_1$, $x'\in\source{\ifk_1}\cap\inter([x,y])$, and $|x'|-|x|=k+1$. We construct an~extraction $\ifk_2\subseteq\ifi_2'$ of $(\ifj_2,\ifi_2')$ by copying.
  More formally, consider $x'\in\source{\ifj_2}$ and let $x=\pre(x')$. If~$\ifk_1$ chooses the $k$th component in $x$ then in $x'$ we choose to $\ifk_2$ the $k$th component as well if it exists, otherwise we
  choose some arbitrary component.

  Let $\pi$ be a~branch where $\ifj$ appears infinitely often, is eventually constant, and $\ifi_2$ has infinite $\liminf$. If~$\ifk_2$ appears infinitely often in~$\pi$, we prove that from the construction above, $f\eqdef \ifk_1^\pre(\pi)$ is asymptotically equivalent to
  $g\eqdef \ifk_2(\pi)$. Let $X\subseteq\nat$. Since~$\ifj$ is eventually constant in~$\pi$ after some point, from the
  construction above, the numbers in $f\restriction_X$ are always smaller than the corresponding numbers in
  $g\restriction_X$. Because $\ifi_2$ tends to infinity, we have that either both $f\restriction_X$ and
  $g\restriction_X$ are bounded or both of them are unbounded. As a~consequence $\ifk_1^\pre(\pi)\sim\ifk_2(\pi)$.

  From the assumptions and the argument above we conclude that:
  \begin{align*}
    \pr\big[\ifk_2\ \io\ \Rightarrow\ \ifk_1^\pre\sim\ifk_2\big]=1,
  \end{align*}
  and hence refute the second statement of the lemma and finish the proof of the converse implication. This
  concludes the proof of Proposition~\ref{pro:constant}.
\subsubsection{Implicit wrappings}
\label{ap:implicit wrappings}  
We demonstrate how to avoid speaking explicitly about the family of intervals $\ifi$ in the formulation of Proposition~\ref{pro:constant}. We will give a proof of Theorem~\ref{thm:const-def} using Proposition~\ref{pro:constant}.
\begin{definition}
\label{def:spaced}
  We say that a~family of intervals~$\ifj$ is \emph{sufficiently spaced} if there exists a~family~$\ifi$ such
  that $\ifj$ wraps $\ifi$ and
  \begin{align*}
    \pr\big[\ifi\ \io\ \Rightarrow\ (\liminf \ifi=\infty)\big]=1.
  \end{align*}
\end{definition}

We claim that
\begin{align*}
  \pr\big[\ifj\ \io\ \Rightarrow\ \text{$\ifj$ is eventually constant}\big]=1,
\end{align*}
is equivalent to
\begin{itemize}
\item[$(\ast)$] $\pr\big[\ifj\ \io\ \Rightarrow\ \limsup \ifj<\infty\big]=1$ and either:\\
  $\pr\big[\ifj\ \io\ \Rightarrow (\lim\ifj = 0)\big]=1$ or\\
  $\pr\big[\ifj\ \io\ \Rightarrow (\liminf\ifj > 0)\big]=1$ and 
  for all $\ifj'\subseteq\ifj$ that are sufficiently spaced we have:
  \begin{align*}
    \pr\big[\ifj'\ \io\ \Rightarrow\  \text{$\ifj'$ is eventually constant}\big]=1.
  \end{align*}
\end{itemize}

Notice that the definition of $\ifj$ wrapping $\ifi$ (see Definition~\ref{def:wrapping}) implicitly implies that all the intervals $[x,y]\in\ifj$ have positive length. However, in the following lemma we prefer to allow the family $\ifj$ to contain some intervals of length~$0$. This explains the additional condition in $(\ast)$.

The forward implication is immediate. For the converse, assume $(\ast)$. Clearly if $\pr\big[\ifj\ \io\ \Rightarrow\ (\lim\ifj =0)\big]=1$ then $\ifj$ is almost surely eventually constant whenever it appears infinitely often.
  
Now suppose towards a contradiction that there is~non\=/zero probability that the following properties hold: $\ifj$ appears infinitely often, is bounded, $[\liminf\ifj>0]$, but $\ifj$ is not eventually constant. In that case, without loss of generality we can assume that $\ifj$ contains no intervals of length~$0$. Then, by Lemma~\ref{lem:d1d2}, there exist $\ell_1>\ell_2\in\nat$ and isolated $\ifj_1\subseteq\ifj$, $\ifj_2\subseteq\ifj$ such that $\ifj_1\subseteq \ifj$ contains intervals of length~$\ell_1$, $\ifj_2\subseteq\ifj$ contains intervals of length~$\ell_2$, and there is~non\=/zero probability that both $\ifj_1$ and $\ifj_2$ appear infinitely often. As $\ifj$ contains no intervals of length~$0$, we know that $\ell_2>0$.

  Take $i=1,2$ and $i'=3-i$ (\ie~the other number). For $k\in\nat$ and $x\in\source{\ifj_i}$ define $S_k(x)$ to be the set of all $x'\in\source{\ifj_{i'}}$ such that :
  \begin{align*}
    &x < x',\text{ and }\\
    &\forall u \in\source{\ifj_{i'}},\  x < u < x'.\ \length([x,u])\leq k{+}\ell_1{+}\ell_2
  \end{align*}
  In other words, $S_k(x)$ contains the first descendants of $x$ in $\source{\ifj_{i'}}$ that are at a~distance at least $k$.

  For all $n\in\nat$ we define $X_n\subseteq\source{\ifj_1}$, $Y_n\subseteq\source{\ifj_2}$ as follows: let $X_0$ be the
  subset of nodes in $\source{\ifj_1}$ that do not have any strict ancestors in $\source{\ifj_1}$ and
  \[
    Y_n=\bigcup_{x\in X_n}S_n(x)\qquad X_{n+1}=\bigcup_{y\in Y_n}S_{n+1}(x).
  \]
  Let $\ifj_1'\subseteq\ifj_1$ (resp. $\ifj_2'\subseteq\ifj_2$) contain all the intervals with sources in $\bigcup_{n\in\nat} X_n$ (resp. in $\bigcup_{n\in\nat} Y_n$). Put $\ifj'=\ifj_1'\cup\ifj_2'$.
  
  \begin{claim}
  For $\ifj'$ defined as above we have $\pr\big[\ifj'\ \io\big]>0$.
  \end{claim}
  
  \begin{proof}
    Directly from the definition, because
    \begin{align*}
      [\ifj_1\ \io]\cap[\ifj_2\ \io]\subseteq [\ifj'\ \io].
    \end{align*}
  \end{proof}

\begin{proposition}
There exists a~family of intervals $\ifi$ such that $\ifj'$ wraps $\ifi$ and for each interval $[x,y]\in\ifj'$ if $x\in X_n\cup Y_n$ then the intervals in $\ifi$ with sources in $\inter([x,y])$ have length exactly $n$. In particular, $\pr\big[\ifi\ \io\ \Rightarrow\ (\lim \ifi=\infty)\big]=1$ and therefore $\ifj'$ is sufficiently spaced.
\end{proposition}

\begin{proof}
It is enough to observe that the intervals added to $\ifi$ by a~naive construction will not overlap with consecutive intervals of $\ifj'$. However, this is guaranteed by the choice of the sets~$S_n(x)$ \added{and the fact that $\ifj'$ contains no trivial intervals}.
\end{proof}

  Finally, if $\pi$ is a~branch in which both $\ifj_1$ and $\ifj_2$ appear infinitely often, then in $\pi$, $\ifj_1'$ and $\ifj_2'$ also appear infinitely often. This implies that
  \begin{align*}
    \pr\big[\ifj'\ \io\ \land\ \text{$\ifj'$ is not eventually constant}\big]>0,
  \end{align*}
  contradicting $(\ast)$.

  Due to the analysis from Appendix~\ref{app:definability}, $(\ast)$ as well as being sufficiently spaced are \msonab definable. Therefore, Theorem~\ref{thm:const-def} follows.


\newcommand{\phiint}{\phi_{\mathrm{int}}}
\newcommand{\phiset}{\phi_{\mathrm{set}}}
\newcommand{\phibnd}{\phi_{\mathrm{bnd}}}
\newcommand{\phiwrp}{\phi_{\mathrm{wrap}}}

\newcommand{\phiLfour}{\phi_{\mathrm{\ref{lem:liminf}}}}
\newcommand{\phiunbnd}{\phi_{\mathrm{ubnd}}}
\newcommand{\phichar}{\phi_{\mathrm{char}}}

\subsection{Definability in \texorpdfstring{\msonab}{MSO+nabla}}
\label{app:definability}

In this technical section we argue why all the properties gradually defined throughout the paper are in fact \msonab definable. Therefore, the section consists of a~pass through the successively defined concepts.

First, as explained in Section~\ref{sec:undecidability}, we will represent a~family of intervals $\ifi$ as a~pair of sets $\src_\ifi=\source{\ifi}$ and $\tgt_\ifi=\target{\ifi}$ of nodes of the tree. Consider the following \mso formulae ($\exists!$~stands for ``there exists a~unique''):
\begin{align*}
\phiint(x,y,\src_\ifi,\tgt_\ifi)&=\ x\in\src_\ifi\land y\in\tgt_\ifi \land x<y \land \\
                                 &\forall z.\ (x<z<y)\Rightarrow z\notin \src_\ifi\land z\notin\tgt_\ifi,
\end{align*}
\begin{align*}
\phiset(\src_\ifi, \tgt_\ifi)=&\ \forall x\in\src_\ifi.\ x\notin\tgt_\ifi\land\\
&\ \forall y\in\tgt_\ifi.\ y\notin\src_\ifi\land\\
 &\ \forall x\in\src_\ifi.\ \exists! y\in\tgt_\ifi.\ \phiint(x,y,\src_\ifi,\tgt_\ifi)\land\\
 &\ \forall y\in\tgt_\ifi.\ \exists! x\in\src_\ifi.\ \phiint(x,y,\src_\ifi,\tgt_\ifi)\land.
\end{align*}
The formula $\phiint(x,y,\src_\ifi,\tgt_\ifi)$ expresses that $[x,y]$ is an~interval in~$\ifi$, while $\phiset(\src_\ifi,\tgt_\ifi)$ means that $(\src_\ifi,\tgt_\ifi)$ in fact represent a~valid family of intervals. Notice that $\ifi\subseteq\ifj$ boils down to saying that $\phiset(\src_\ifi, \tgt_\ifi)$, $\phiset(\src_\ifj, \tgt_\ifj)$, and $\src_\ifi\subseteq\src_\ifj$ and $\tgt_\ifi\subseteq\tgt_\ifj$.

\begin{remark}
Consider a~representation $(\src_\ifi,\tgt_\ifi)$ of a~family of intervals $\ifi$. Let $X$ be a~set of nodes and $\pi$ be a~branch (also represented as a~set of nodes). Then the following conditions are \mso definable: $X$\ \fo in $\pi$; $X$\ \io in $\pi$; $\ifi$\ \io in $\pi$.
\end{remark}

Using the above remark, the $(\ast)$ property of Lemma~\ref{lem:liminf} is easily \msonab definable by the following formula
\begin{align*}
\phiLfour(\ifi)\eqdef\ & \exists \ifi'\subseteq \ifi.\\
&\quad \phiset(\src_{\ifi'},\tgt_{\ifi'}) \wedge \\
&\quad \nabla \pi.\ \src_{\ifi'}\ \text{\io in $\pi$}\, \wedge\\
&\quad \forall \ifj\subseteq\ifi'.\  \phiset(\src_{\ifj},\tgt_{\ifj})\Rightarrow \\
&\qquad \nabla \pi.\ \big(\text{$\src_{\ifj}$\ \io in $\pi$}\Leftrightarrow \text{$\tgt_{\ifj}$\ \io in $\pi$}\big).
\end{align*}

To define in \msonab the property \eqref{eq:unbounded formula}, one uses the negation of the condition from Lemma~\ref{lem:liminf}: it is equivalent to saying that $\pr\big[\ifi\ \io\Rightarrow\ (\liminf\ifi = \infty)\big]=1$. This means that the following formula is equivalent to saying that $\ifi$ is unbounded
\begin{align*}
\phiunbnd(\ifi)\eqdef\ & \exists \ifj\subseteq \ifi.\\
&\quad \phiset(\src_{\ifj},\tgt_{\ifj}) \wedge \\
&\quad \nabla \pi.\ \big(\text{$\src_{\ifi}$\ \io in $\pi$}\Leftrightarrow\text{$\src_{\ifj}$\ \io in $\pi$}\big) \wedge\\
&\quad \lnot \phiLfour(\ifj).
\end{align*}

Thus, using Lemma~\ref{lem:char}, a~characteristic of a~family of intervals is \msonab definable by the following formula:
\begin{align*}
\phichar(\src_\ifi, \tgt_\ifi, X)&\eqdef \phiset(\src_\ifi, \tgt_\ifi)\wedge\\
&\exists \ifj\subseteq \ifi.\\
&\quad \phiset(\src_{\ifj},\tgt_{\ifj}) \wedge \\
&\quad \phiunbnd(\ifj) \wedge \\
&\quad \nabla \pi.\ \big(\text{$X$\ \io in $\pi$}\Leftrightarrow\text{$\src_{\ifj}$\ \io in $\pi$}\big) \wedge\\
&\quad \forall \ifj'\subseteq\ifi.\\
&\qquad  (\phiset(\src_{\ifj},\tgt_{\ifj})\land \phiunbnd(\ifj))\Rightarrow \\
&\qquad \nabla \pi.\ \big(\text{$\src_{\ifj'}$\ \io in $\pi$}\Rightarrow\text{$\src_{\ifj}$\ \io in $\pi$}\big).
\end{align*}

\begin{remark}
\label{rem:bnd-def}
From that moment on we will represent families of intervals $\ifi$ as~triples $\src_\ifi, \tgt_\ifi, X_\ifi$, where $\phichar(\src_\ifi, \tgt_\ifi, X_\ifi)$ holds. Thanks to that representation, we have
\begin{align*}
  \pr\big[X_\ifi\ \io \iff \limsup\ifi=\infty\big]=1.
\end{align*}
Therefore, (up to a~set of branches of measure $0$) ``$\limsup\ifi(\pi)<\infty$'' is \msonab definable by the formula:
\begin{align*}
  &\phibnd(\ifi, \pi)\eqdef\\
  &\qquad\exists x\in\pi.\ \forall y\in\pi.\ (x<y)\Rightarrow y\notin X_\ifi.
\end{align*}
\end{remark}

Notice that Definition~\ref{def:precedes} is already stated as an~\mso property. Moreover, the relation between $x$ and $x'$ given by the function $\pre$ is also \mso definable directly from the definition. This leads to the conclusion that one can define in \mso that $\ifi_1$ precedes $\ifi_2$ and $\ifi'$ is the effect of applying the function $\pre$ (resp. $\suc$) to a~family of intervals $\ifi\subseteq\ifi_2$ (resp. $\ifi\subseteq\ifi_1$). Clearly the fact that $\ifi_1$ and $\ifi_2$ are isolated is also \mso definable using our encoding.

Remark~\ref{rem:bnd-def} immediately implies that the $(\ast)$ property in the proof of Lemma~\ref{lem:c1pre} in Appendix~\ref{app:lemc1pre} is \msonab definable---it is enough to replace each occurrence of $[\limsup\ifi<\infty]$ by~$\phibnd$.

Next, we investigate the properties from Section~\ref{sec:wrappings}. It is easy to see that the following formula defines that $\ifj$ wraps $\ifi$:
\begin{align*}
\phiwrp\eqdef&\ \phiset(\ifi)\land \phiset(\ifj)\land \phi_{\geq 1}(\ifj)\land \\
       &\ \forall x'.\ x'\in\src_\ifi \Leftrightarrow \exists x,y.\ \phiint(x,y,\src_\ifj,\tgt_\ifj)\\
       &\qquad\land x < x' < y, 
\end{align*}
where $\phi_{\geq 1}(\ifj)$ states that every interval in $\ifj$ has length at least 1. 

Further, Definition~\ref{def:tail-prec} is itself expressed in \mso. The same holds for the notion of \emph{extraction} and the order $\ifi_1\leq\ifi_2$. These observations give us sufficient background to study the $(\ast)$ property of Proposition~\ref{pro:constant}. The only part of this property that is not directly \msonab formalisable is~\eqref{eq:e-not-sim-e}. However, under the previous assumptions of the formula, $\ifk_1$ and $\ifk_2$ satisfy the conditions of Lemma~\ref{lem:c1pre}, and therefore~\eqref{eq:e-not-sim-e} is in fact \msonab definable.

Following the construction from the main body, observe that being \emph{sufficiently spaced} (see Definition~\ref{def:spaced}) is definable in \msonab. This is because the requirement $\pr\big[\ifi\ \io\ \Rightarrow (\liminf \ifi=\infty)\big]=1$ is just the negation of the first statement of Lemma~\ref{lem:liminf}, \ie~it is expressible by the formula $\lnot\phiLfour(\ifi)$. Thus, the $(\ast)$ property discussed in Section~\ref{ap:implicit wrappings} is also \msonab definable. Therefore, Theorem~\ref{thm:const-def} follows.


\end{document}
